\documentclass{llncs}

\usepackage{multicol}

\usepackage{amsmath}
\usepackage{amsfonts}
\usepackage{amssymb}
\usepackage{verbatim}

\usepackage{graphicx}
\usepackage{epic}
\usepackage{eepic}
\usepackage{epsfig,float}
\usepackage{pdfsync}

\usepackage{multicol}
\renewcommand{\le}{\leqslant}
\renewcommand{\ge}{\geqslant}

\newcommand{\ol}{\overline}
\newcommand{\eps}{\varepsilon}
\newcommand{\emp}{\emptyset}

\newcommand{\Sig}{\Sigma}
\newcommand{\sig}{\sigma}
\newcommand{\noin}{\noindent}

\newcommand{\bi}{\begin{itemize}}
\newcommand{\ei}{\end{itemize}}
\newcommand{\be}{\begin{enumerate}}
\newcommand{\ee}{\end{enumerate}}
\newcommand{\bd}{\begin{description}}
\newcommand{\ed}{\end{description}}
\newcommand{\bq}{\begin{quote}}
\newcommand{\eq}{\end{quote}}

\newcommand{\cA}{{\mathcal A}}
\newcommand{\cB}{{\mathcal B}}

\newcommand{\cD}{{\mathcal D}}

\newcommand{\cN}{{\mathcal N}}
\newcommand{\cP}{{\mathcal P}}

\newcommand{\cR}{{\mathcal R}}

\newcommand{\cT}{{\mathcal T}}

\newcommand{\one}{{\mathbf 1}}

\newcommand{\lraL}{{\hspace{.1cm}{\leftrightarrow_L} \hspace{.1cm}}}

\newcommand{\defeq}{\stackrel{\rm def}{=}}

\newcommand{\rev}{R}
\newcommand{\deter}{D}

\spnewtheorem{conj}{Conjecture}{\bfseries}{\rmfamily}

\title{Most Complex Regular Right-Ideal Languages\thanks{This work was supported by the Natural Sciences and Engineering Research Council of Canada under grant No.~OGP0000871.}
}

\author{Janusz~Brzozowski and Gareth Davies}

\authorrunning{Brzozowski, Davies}   

\institute{David R. Cheriton School of Computer Science, University of Waterloo, \\
Waterloo, ON, Canada N2L 3G1\\
{\tt \{brzozo,gdavies\}@uwaterloo.ca}
}
\begin{document}

\maketitle
\begin{abstract}
A right ideal is a language $L$ over an alphabet $\Sig$ that satisfies $L=L\Sig^*$.
We show that there exists a stream (sequence) ($R_n \mid n \ge 3$) of regular right ideal languages, where $R_n$ has $n$ left quotients and is most complex under the following measures of complexity: the state complexities of the left quotients,  the number of atoms (intersections of complemented and uncomplemented left quotients), the state complexities of the atoms, the size of the syntactic semigroup, the state complexities of the operations of reversal, star, and product, and the state complexities of all binary boolean operations.
In that sense, this stream of right ideals is a universal witness.

\medskip

\noin
{\bf Keywords:}
atom,   operation,  quotient, regular language, right ideal,  state complexity,  syntactic semigroup, universal witness

\end{abstract}

\section{Introduction}
\label{sec:introduction}

Brzozowski introduced a list of conditions that a complex regular language should satisfy and found a ``universal witness'' that meets all these  conditions~\cite{Brz13}. This witness meets the upper bounds for complexity of atoms and for all the basic operations: reverse, star, boolean operations, product (concatenation, catenation), as well as a large number of combined operations. However, it does not work for subclasses of regular languages, since it generally lacks the properties of those classes.

This paper is a case study that investigates whether the approach used for general regular languages can be extended to subclasses. We present a universal witness for regular right ideals and  show that it has maximally complex atoms for right ideals and meets the bounds for all basic operations on right ideals.

For a further discussion of regular right ideals see~\cite{BJL13,BrYe11}.
It is pointed out in~\cite{BJL13} that 
right ideals deserve to be studied for several reasons: 
They are fundamental objects in semigroup theory, they appear in the theoretical computer science literature 
as early as 1965,
and continue to be of interest in the present.
Right ideal languages are complements 
of prefix-closed languages, and are closed with respect to 
the ``has a word as a prefix'' relation. 
They are special cases of convex languages, 
which form a much larger class. 
Finally, besides being of theoretical interest, right ideals also play a role in algorithms for pattern matching: 
When searching for all words beginning in a word from some set $L$, one is looking for all the words of the right ideal $L\Sig^*$.

\section{Background}

A \emph{deterministic finite automaton (DFA)} $\cD= (Q,\Sig,\delta,q_1,F)$ consists of 
a finite non-empty set $Q$ of \emph{states},
a finite non-empty \emph{alphabet} $\Sig$,  
a \emph{transition function} $\delta\colon Q\times \Sig\to Q$, an
\emph{initial state} $q_1\in Q$, and 
a set $F\subseteq Q$ of \emph{final states}.
The transition function is extended to functions $\delta'\colon Q \times \Sig^*\to Q$ and $\delta''\colon 2^Q \times \Sig^*\to 2^Q$ as usual, but these extensions are  also denoted by $\delta$. A state $q$ of a DFA is \emph{reachable} if there is a word $w\in\Sig^*$ such that $\delta(q_1,w)=q$. 
The \emph{language  accepted} by $\cD$ is $L(\cD)=\{w\in\Sig^* \mid \delta(q_1,w)\in F\}$.
Two DFAs are \emph{equivalent} if their languages are the same.
The \emph{language of a state} $q$ is the language accepted by the DFA $\cD_q= (Q,\Sig,\delta,q,F)$.
Two states are \emph{equivalent} if their languages are equal; otherwise, they are \emph{distinguishable} by some word that is in the language of one of the states, but not of the other. 
If $S \subseteq Q$, two states $p,q \in Q$ are \emph{distinguishable with respect to $S$} if there is a word $w$ such that $\delta(p,w) \in S$ and $\delta(q,w) \not\in S$.
A~DFA is \emph{minimal} if all of its states are reachable and no two states are equivalent.
A state is \emph{empty} if its language is empty.

A \emph{nondeterministic finite automaton (NFA)} is a tuple $\cN= (Q,\Sig,\eta,Q_1,F)$, where 
$Q$, $\Sig$, and $F$ are as in a DFA, $\eta\colon Q\times \Sig\to 2^Q$ is the transition function 
and $Q_1\subseteq Q$ is the \emph{set of initial states}.
An $\eps$-NFA has all the features of an NFA but its transition function 
$\eta\colon Q\times (\Sig\cup \{\eps\})\to 2^Q$ allows also transitions under the empty word. 
The \emph{language accepted} by an NFA or an $\eps$-NFA is the set of words $w$ for which there exists a sequence of transitions such that the concatenation of the symbols causing the transitions is $w$,
and this sequence leads from a state in $Q_1$ to a state in $F$.
Two NFAs are \emph{equivalent} if they accept the same language.

We use the following operations on automata: 
\be
\item
The \emph{determinization} operation $\deter$ 
applied to an NFA $\cN$ yields a DFA $\cN^{\deter}$ obtained by 
the 
subset construction, where only subsets  reachable 
from the initial subset of $\cN^\deter$ are used and the empty subset, 
if present, is included. 
\item
The \emph{reversal} operation $\rev$ applied to an NFA $\cN$ yields 
an NFA $\cN^{\rev}$, where sets of initial and final states of $\cN$ are 
interchanged and each transition 
is reversed. 
\ee

Let $\cD = (Q, \Sig, \delta, q_1, F)$ be a DFA. For each word $w \in \Sig^*$, the transition function induces a transformation $t_w$ of $Q$ by  $w$: for all $q \in Q$, 
$qt_w \defeq \delta(q, w).$ 
The set $T_{\cD}$ of all such transformations by non-empty words forms a semigroup of transformations called the \emph{transition semigroup} of $\cD$~\cite{Pin97}. 
Conversely, we can use a set  $\{t_a \mid a \in \Sig\}$ of transformations to define $\delta$, and so also the DFA $\cD$. We write $a\colon t$, where $t$ is a transformation of $Q$, to mean that the transformation induced by $a \in \Sig$ is~$t$.

The \emph{Myhill congruence}~\cite{Myh57} $\lraL$ of a language $L\subseteq \Sig^*$ is defined on $\Sig^+$ as follows:
\begin{equation*}
\mbox{For $x, y \in \Sig^+$, } x \lraL y \mbox{ if and only if } uxv\in L  \Leftrightarrow uyv\in L\mbox { for all } u,v\in\Sig^*.
\end{equation*}
This congruence is also known as the \emph{syntactic congruence} of $L$.
The quotient set $\Sig^+/ \lraL$ of equivalence classes of the relation $\lraL$ is a semigroup called the \emph{syntactic semigroup} of $L$.
If  $\cD$ is the minimal DFA of $L$, then $T_{\cD}$ is isomorphic to the syntactic semigroup $T_L$ of $L$~\cite{Pin97}, and we represent elements of $T_L$ by transformations in~$T_{\cD}$. 
An arbitrary transformation can be written in the form
\begin{equation*}\label{eq:transmatrix}
t=\left( \begin{array}{ccccc}
q_1 & q_2 &   \cdots &  q_{n-1} & q_n \\
p_1 & p_2 &   \cdots &  p_{n-1} & p_n
\end{array} \right ),
\end{equation*}
where $p_k = kt$,  $1\le k\le n$, and $p_k\in Q$. 

A \emph{permutation} of $Q$ is a mapping of $Q$ \emph{onto} itself. 
The \emph{identity} transformation $\one$ maps  each element to itself, that is, $q\one=q$ for $q\in Q$.
A transformation $t$ is a \emph{cycle} of length $k$ if there exist pairwise different elements $p_1,\ldots,p_k$ such that
$p_1t=p_2,p_2t=p_3,\ldots, p_{k-1}t=p_k$,  $p_kt=p_1$, and other elements of $Q$ are mapped to themselves.
A cycle is denoted by $(p_1,p_2,\ldots,p_k)$.
A \emph{transposition} is a cycle $(p,q)$.
A~\emph{unitary} transformation, denoted by $(p\rightarrow q)$, has $pt=q$ and $rt=r$ for all $r\neq p$.

The set of all permutations of a set $Q$ of $n$ elements is a group, 
called the \emph{symmetric group} of degree $n$. 
It is well known that  two generators are sufficient to generate the symmetric group of degree $n$. 
Without loss of generality, from now on  we assume that $Q=\{1,2,\dots, n\}$.
\begin{proposition}[Permutations]
\label{prop:piccard}
The symmetric group of size $n!$ can be generated by any cyclic
permutation of $n$ elements together with any transposition. In particular, it can be generated by
$(1,2,\ldots, n)$ and  $(1,2)$.
\end{proposition}

The set of all transformations of a set $Q$,
denoted by $\cT_Q$, 
is a semigroup, in fact a monoid with $\one$ as the identity. 
It is well known that three transformations of $Q$ are sufficient to generate the semigroup $\cT_Q$,  and fewer than three generators are insufficient.

\begin{proposition}[Transformations]
\label{prop:piccard2}
The transformation monoid $\cT_Q$ of size $n^n$ can be generated by any cyclic
permutation of $n$ elements together with any transposition and any unitary transformation. In particular, $\cT_Q$ can be generated by $c\colon (1,2,\ldots, n)$,  $t\colon (1,2)$ and ${r\colon (n \rightarrow 1)}$.
\end{proposition}

The \emph{state complexity of a regular language}~\cite{Yu01} $L$ over a finite alphabet $\Sig$ is the number of states in the minimal DFA recognizing $L$. An equivalent notion is that of \emph{quotient complexity}~\cite{Brz10}, which is  the number of distinct left quotients of $L$, where
 the left quotient of $L\subseteq \Sig^*$ by a word $w\in \Sig^*$ is the language 
$w^{-1}L=\{x\in\Sig^*\mid wx\in L\}$. This paper uses  \emph{complexity} for both of these equivalent notions, and this  term will not be used for any other property here. 

The \emph{(state/quotient) complexity of an operation} on regular languages is the maximal  complexity of the language resulting from the operation as a function of the  complexities of the arguments.
For example, for $L \subseteq \Sig^*$, the complexity of the reverse $L^R$ of $L$ is $2^n$ if the complexity of $L$ is $n$, since a minimal DFA for $L^R$ can have at most $2^n$ states and there exist languages meeting this bound~\cite{Mir66}.

There are two parts to the process of  establishing the  complexity of an operation.
First, one must find an \emph{upper bound} on the  complexity of the result of the operation
by using quotient computations or automaton constructions.
Second, one must find \emph{witnesses}  that meet this upper bound.
One usually defines a sequence $(L_n\mid n\ge k)$ of languages, where  $k$ is some small positive integer.  This sequence will be called a \emph{stream}. The languages  in a stream differ only in the parameter $n$. 
For example, 
one might study unary languages $(\{a^n\}^*\mid n\ge 1)$ that have zero $a$'s modulo $n$. 
A unary operation takes its argument from a stream $(L_n\mid n\ge k)$.
For a binary operation, one adds  a stream $(K_n\mid n\ge k)$ as the second argument. While the witness streams are normally different for different operations, the main result of this paper shows that a single stream can meet the complexity bounds for all operations in the case of right ideals.

Atoms of regular languages were introduced in 2011~\cite{BrTa11}, and their complexities were studied  in 2012~\cite{BrTa12}.
Let $L$ be a regular language with quotients $K = \{K_1,\dotsc,K_n\}$. Each subset $S$ of $K$ defines an \emph{atomic intersection} $A = \widetilde{K_1} \cap \dotsb \cap \widetilde{K_n}$, where $\widetilde{K_i}$ is $K_i$ if $K_i \in S$ and $\ol{K_i}$ otherwise. An \emph{atom} of $L$ is a non-empty atomic intersection. Since non-empty atomic intersections are pairwise disjoint, every atom $A$ has a unique atomic intersection associated with it, and this atomic intersection has a unique subset $S$ of $K$ associated with it. This set $S$ is called the \emph{basis} of $A$ and is denoted by $\cB(A)$. The \emph{co-basis} of $A$ is $\ol{\cB(A)} = K \setminus \cB(A)$. The basis of an atom is the set of quotients of $L$ that occur uncomplemented as terms of the corresponding intersection, and the co-basis is the set of quotients that occur complemented.

It was proven in~\cite{BrTa11} that each regular language $L$ defines a unique set of atoms, and that every quotient of $L$ (including $L$ itself) and every quotient of every atom of $L$ is a union of atoms. Thus the atoms of $L$ are its basic building blocks. In~\cite{Brz13} it was argued that it is useful to consider the complexity of a language's atoms when searching for ``most complex'' regular languages, since one would expect a complex language to have  complex building blocks.

Let ${\bf A} = \{A_1,\dotsc,A_m\}$ be the set of atoms of $L$. The \emph{\'atomaton} of $L$ is the NFA $\cA = ({\bf A},\Sig,\eta,{\bf A}_I,A_f)$, where the \emph{initial} atoms are ${\bf A}_I = \{ A_i \mid L \in \cB(A_i) \}$, the \emph{final} atom $A_f$ is the unique atom such that $K_i \in \cB(A_f)$ if and only if $\eps \in K_i$, and $A_j \in \eta(A_i,a)$ if and only if $aA_j \subseteq A_i$. The \'atomaton has the property that each state is its own language, that is, the language of the state $A$ of $\cA$ is the atom $A$ of $L$. Also, since each regular language defines a unique set of atoms, each regular language also defines a unique \'atomaton.

It was shown in~\cite{BrTa11,BrTa12}  that if $\cD$ is the minimal DFA for $L$, then $\cA^\rev$ is a minimal DFA 
that accepts $L^R$, and $\cA^\rev$ is isomorphic to $\cD^{\rev\deter}$. From this it follows that $\cA$ is isomorphic to the NFA $\cD^{\rev\deter\rev}$. In particular, we have the following isomorphism:
\begin{proposition}[\'Atomaton Isomorphism]
\label{prop:atomiso}
Let $L$ be a regular language with quotients $K = \{K_1,\dotsc,K_n\}$ and set of atoms ${\bf A}$. Let $\cD$ be the minimal DFA of $L$, with state set $Q = \{1,\dotsc,n\}$ such that the language of state $i$ is $K_i$. Then the map $\varphi \colon  {\bf A} \rightarrow 2^Q$ defined by $\varphi(A) = \{ i \mid K_i \in \cB(A) \}$ is an isomorphism between $\cA$ and $\cD^{\rev\deter\rev}$.
\end{proposition}

\section{Main Results}
The right ideal stream $(R_n \mid n \ge 3)$ that turns out to be most complex is   defined as follows:

\begin{definition}
\label{def:mostcomplex}
For $n\ge 3$, let $\cR_n=\cR_n(a,b,c,d)=(Q,\Sig,\delta, 1, F)$, where 
$Q=\{1,\dots,n\}$
is the set of states\footnote{Although $Q$, $\delta$, and $F$ depend on $n$,  this dependence is not shown to keep the notation as simple as possible.}, 
$\Sig=\{a,b,c,d\}$ is the alphabet, 
the transformations defined by $\delta$ are
$a\colon (1,\dots,n-1)$,
$b\colon(2,\ldots,n-1)$,
${c\colon(n-1\rightarrow 1)}$
and ${d\colon(n-1\rightarrow n)}$, 
$1$ is the initial state, and
$F=\{n\}$ is the set of final states.
Let $R_n=R_n(a,b,c,d)$ be the language accepted by~$\cR_n$.
\end{definition}

\begin{figure}[t]
\begin{center}
\setlength{\unitlength}{0.00043745in}
\begingroup\makeatletter\ifx\SetFigFont\undefined%
\gdef\SetFigFont#1#2#3#4#5{%
  \reset@font\fontsize{#1}{#2pt}%
  \fontfamily{#3}\fontseries{#4}\fontshape{#5}%
  \selectfont}%
\fi\endgroup%
{\renewcommand{\dashlinestretch}{30}
\begin{picture}(7905,2095)(0,-10)
\put(7422,1107){\makebox(0,0)[lb]{\smash{{\SetFigFont{6}{7.2}{\familydefault}{\mddefault}{\updefault}$n$}}}}
\put(2927.000,1592.333){\arc{333.333}{2.2143}{7.2105}}
\blacken\path(3064.638,1555.417)(3027.000,1459.000)(3105.107,1526.913)(3064.638,1555.417)
\put(1847.000,1592.333){\arc{333.333}{2.2143}{7.2105}}
\blacken\path(1984.638,1555.417)(1947.000,1459.000)(2025.107,1526.913)(1984.638,1555.417)
\put(5367.000,1562.333){\arc{333.333}{2.2143}{7.2105}}
\blacken\path(5504.638,1525.417)(5467.000,1429.000)(5545.107,1496.913)(5504.638,1525.417)
\put(7542.000,1652.333){\arc{333.333}{2.2143}{7.2105}}
\blacken\thicklines
\path(7673.377,1655.553)(7642.000,1519.000)(7739.310,1619.806)(7673.377,1655.553)
\thinlines
\put(771,1144){\ellipse{630}{630}}
\put(1865,1142){\ellipse{630}{630}}
\put(6455,1144){\ellipse{630}{630}}
\put(5367,1144){\ellipse{630}{630}}
\put(7537,1153){\ellipse{720}{720}}
\put(2935,1144){\ellipse{630}{630}}
\put(7534,1152){\ellipse{630}{630}}
\path(2172,1144)(2622,1144)
\blacken\thicklines
\path(2487.000,1106.500)(2622.000,1144.000)(2487.000,1181.500)(2487.000,1106.500)
\thinlines
\path(3252,1144)(3702,1144)
\blacken\thicklines
\path(3567.000,1106.500)(3702.000,1144.000)(3567.000,1181.500)(3567.000,1106.500)
\thinlines
\path(1092,1144)(1542,1144)
\blacken\thicklines
\path(1407.000,1106.500)(1542.000,1144.000)(1407.000,1181.500)(1407.000,1106.500)
\thinlines
\path(6769,1144)(7167,1144)
\blacken\thicklines
\path(7032.000,1106.500)(7167.000,1144.000)(7032.000,1181.500)(7032.000,1106.500)
\thinlines
\path(4602,1152)(5052,1152)
\blacken\thicklines
\path(4917.000,1114.500)(5052.000,1152.000)(4917.000,1189.500)(4917.000,1114.500)
\thinlines
\path(5675,1144)(6125,1144)
\blacken\thicklines
\path(5990.000,1106.500)(6125.000,1144.000)(5990.000,1181.500)(5990.000,1106.500)
\thinlines
\path(12,1144)(462,1144)
\blacken\thicklines
\path(327.000,1106.500)(462.000,1144.000)(327.000,1181.500)(327.000,1106.500)
\thinlines
\path(6192,949)(6191,949)(6190,948)
	(6187,947)(6182,946)(6175,943)
	(6166,940)(6154,937)(6140,932)
	(6122,926)(6101,919)(6077,912)
	(6051,903)(6021,894)(5988,883)
	(5952,872)(5914,860)(5873,848)
	(5830,835)(5785,822)(5737,808)
	(5688,794)(5638,779)(5586,765)
	(5532,750)(5478,736)(5422,722)
	(5365,707)(5307,693)(5248,679)
	(5188,666)(5127,653)(5064,640)
	(5000,627)(4935,615)(4868,603)
	(4800,592)(4730,581)(4657,571)
	(4583,561)(4507,552)(4429,544)
	(4349,536)(4267,530)(4183,524)
	(4098,519)(4013,516)(3927,514)
	(3830,513)(3735,514)(3643,517)
	(3554,521)(3469,526)(3387,533)
	(3310,540)(3235,549)(3165,558)
	(3097,568)(3033,579)(2971,591)
	(2912,603)(2855,615)(2800,629)
	(2747,642)(2696,656)(2647,671)
	(2599,685)(2553,700)(2508,715)
	(2466,730)(2425,744)(2386,759)
	(2349,773)(2314,786)(2281,799)
	(2251,812)(2223,823)(2198,834)
	(2176,843)(2156,852)(2139,859)
	(2125,865)(2114,871)(2105,875)
	(2099,878)(2089,882)
\blacken\thicklines
\path(2228.272,866.680)(2089.000,882.000)(2200.417,797.044)(2228.272,866.680)
\thinlines
\path(6319,852)(6318,852)(6317,851)
	(6314,849)(6310,847)(6304,844)
	(6296,839)(6285,834)(6272,826)
	(6255,818)(6236,808)(6214,796)
	(6189,783)(6160,768)(6129,752)
	(6095,734)(6057,716)(6017,695)
	(5975,674)(5929,652)(5882,629)
	(5832,605)(5780,580)(5726,555)
	(5671,530)(5614,504)(5555,478)
	(5495,452)(5434,427)(5372,401)
	(5308,375)(5244,350)(5178,326)
	(5111,301)(5043,278)(4974,254)
	(4903,232)(4831,210)(4757,189)
	(4682,169)(4605,149)(4527,131)
	(4446,113)(4364,96)(4279,81)
	(4192,67)(4103,54)(4012,43)
	(3919,33)(3824,25)(3727,18)
	(3629,14)(3531,12)(3432,12)
	(3326,15)(3222,20)(3119,28)
	(3020,38)(2924,50)(2830,63)
	(2741,79)(2654,96)(2571,114)
	(2491,134)(2413,155)(2339,177)
	(2268,200)(2198,224)(2131,249)
	(2067,274)(2004,300)(1943,327)
	(1884,355)(1826,383)(1770,412)
	(1716,441)(1663,470)(1611,499)
	(1561,528)(1513,558)(1466,587)
	(1421,616)(1378,644)(1336,671)
	(1297,698)(1260,724)(1225,748)
	(1192,771)(1162,793)(1134,813)
	(1109,831)(1087,848)(1068,862)
	(1051,875)(1036,886)(1024,895)
	(1015,903)(1007,908)(1002,913)(994,919)
\blacken\thicklines
\path(1124.500,868.000)(994.000,919.000)(1079.500,808.000)(1124.500,868.000)
\put(716,1077){\makebox(0,0)[lb]{\smash{{\SetFigFont{7}{8.4}{\rmdefault}{\mddefault}{\updefault}$1$}}}}
\put(1804,1077){\makebox(0,0)[lb]{\smash{{\SetFigFont{7}{8.4}{\rmdefault}{\mddefault}{\updefault}$2$}}}}
\put(2891,1077){\makebox(0,0)[lb]{\smash{{\SetFigFont{7}{8.4}{\rmdefault}{\mddefault}{\updefault}$3$}}}}
\put(6190,1081){\makebox(0,0)[lb]{\smash{{\SetFigFont{6}{7.2}{\familydefault}{\mddefault}{\updefault}$n-1$}}}}
\put(1137,1279){\makebox(0,0)[lb]{\smash{{\SetFigFont{7}{8.4}{\familydefault}{\mddefault}{\updefault}$a$}}}}
\put(5097,1099){\makebox(0,0)[lb]{\smash{{\SetFigFont{6}{7.2}{\familydefault}{\mddefault}{\updefault}$n-2$}}}}
\put(1684,1857){\makebox(0,0)[lb]{\smash{{\SetFigFont{7}{8.4}{\familydefault}{\mddefault}{\updefault}$c,d$}}}}
\put(2697,1857){\makebox(0,0)[lb]{\smash{{\SetFigFont{7}{8.4}{\familydefault}{\mddefault}{\updefault}$c,d$}}}}
\put(5089,1849){\makebox(0,0)[lb]{\smash{{\SetFigFont{7}{8.4}{\familydefault}{\mddefault}{\updefault}$c,d$}}}}
\put(7129,1909){\makebox(0,0)[lb]{\smash{{\SetFigFont{7}{8.4}{\familydefault}{\mddefault}{\updefault}$a,b,c,d$}}}}
\put(4055,1085){\makebox(0,0)[lb]{\smash{{\SetFigFont{7}{8.4}{\familydefault}{\mddefault}{\updefault}$\cdots$}}}}
\put(6912,1264){\makebox(0,0)[lb]{\smash{{\SetFigFont{7}{8.4}{\familydefault}{\mddefault}{\updefault}$d$}}}}
\put(2181,1264){\makebox(0,0)[lb]{\smash{{\SetFigFont{7}{8.4}{\familydefault}{\mddefault}{\updefault}$a,b$}}}}
\put(3254,1264){\makebox(0,0)[lb]{\smash{{\SetFigFont{7}{8.4}{\familydefault}{\mddefault}{\updefault}$a,b$}}}}
\put(4559,1279){\makebox(0,0)[lb]{\smash{{\SetFigFont{7}{8.4}{\familydefault}{\mddefault}{\updefault}$a,b$}}}}
\put(5684,1271){\makebox(0,0)[lb]{\smash{{\SetFigFont{7}{8.4}{\familydefault}{\mddefault}{\updefault}$a,b$}}}}
\put(3386,101){\makebox(0,0)[lb]{\smash{{\SetFigFont{7}{8.4}{\familydefault}{\mddefault}{\updefault}$a,c$}}}}
\put(3791,589){\makebox(0,0)[lb]{\smash{{\SetFigFont{7}{8.4}{\familydefault}{\mddefault}{\updefault}$b$}}}}
\put(462,1864){\makebox(0,0)[lb]{\smash{{\SetFigFont{7}{8.4}{\familydefault}{\mddefault}{\updefault}$b,c,d$}}}}
\thinlines
\put(772.000,1577.333){\arc{333.333}{2.2143}{7.2105}}
\blacken\path(909.638,1540.417)(872.000,1444.000)(950.107,1511.913)(909.638,1540.417)
\end{picture}
}
\end{center}
\caption{Automaton $\cR_n$ of a most complex right ideal $R_n$.} 
\label{fig:RightIdeal}
\end{figure}
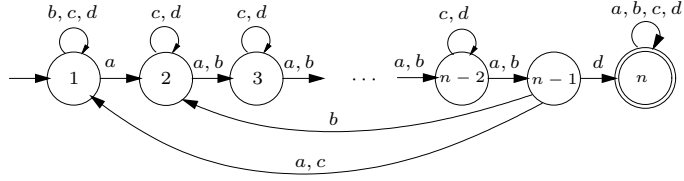

The structure of the DFA $\cR_n(a,b,c,d)$ is shown in Figure~\ref{fig:RightIdeal}. 
Note that  input $b$ induces the identity transformation in $\cR_n$ for $n=3$.

It is worth noting this stream of languages is very similar to the stream $( L_n \mid n \ge 2)$  defined in~\cite{BrTa12} and shown to be a ``universal witness'' in~\cite{Brz13}. 
In that stream, $L_n$ is defined by the DFA $\cD_n=\cD_n(a,b,c)=(Q,\Sig,\delta, 1, \{n\})$, where 
$Q=\{1,\dots,n\}$,
$\Sig=\{a,b,c\}$,  and
$\delta$ is defined by
$a\colon (1,\dots,n-1)$,
$b\colon(1,2)$, and
${c\colon(n-1\rightarrow 1)}$.
The automaton $\cR_n$ can be constructed by taking $\cD_{n-1}$, adding a new state $n$ and a new input $d\colon  (n-1 \rightarrow n)$, making $n$ the only final state, and having $b$ induce the cyclic permutation $(2,\dotsc,n-1)$ (rather than  $(1,2)$). The new state and input are necessary to ensure $R_n$ is a right ideal for all $n$. Changing the transformation induced by $b$ is necessary since, if $b$ were to induce $(1,2)$ in $\cR_n$, then $R_n$ would not meet the bound for product.

We can generalize this definition to a stream $(R_n \mid n \ge 1)$ by noting that
when $n=1$, all four inputs induce the identity transformation, and when $n = 2$, $a$, $b$ and $c$ induce the identity transformation while $d$ induces $(1 \rightarrow 2)$. Hence $R_1 = \{a,b,c,d\}^*$ and $R_2 = \{a,b,c\}^*d\{a,b,c,d\}^*$. However, the complexity bound for star is not reached by $R_1$ and the complexity bounds for boolean operations are not reached when one of the operands is $R_1$ or $R_2$. Thus we require $n \ge 3$ for $R_n$ to be a true universal witness.

In some cases, the complexity bounds can be reached even when the alphabet size is reduced. If the letter $c$ is not needed, then we let $\cR_n(a,b,d)$ be the DFA of Definition~\ref{def:mostcomplex} restricted to inputs $a$, $b$ and $d$, and let $R_n(a,b,d)$ be the language recognized by this DFA.
If both $b$ and $c$ are not needed, we use $\cR_n(a,d)$  and $R_n(a,d)$.
We also define $\cR_n(b,a,d)$ to be the DFA obtained from $\cR_n(a,b,d)$ by interchanging the roles of the inputs $a$ and $b$, and let $R_n(b,a,d)$ be the corresponding language.

\begin{theorem}[Main Results]
The language $R_n=R_n(a,b,c,d)$ has the properties listed below. Moreover, all the complexities of $R_n$ are the maximal possible for right ideals. The results hold for all $n \ge 1$ unless otherwise specified.
\bi
\item
$R_n(a,d)$ has $n$ quotients, that is, its (state/quotient) complexity is $n$.
\item
The syntactic semigroup of $R_n(a,b,c,d)$ has cardinality $n^{n-1}$.
\item
Quotients of $R_n(a,d)$ have complexity $n$, except for the quotient $\{a,d\}^*$, which has complexity 1.
\item
$R_n(a,b,c,d)$ has $2^{n-1}$ atoms.
\item
The atom of $R_n(a,b,c,d)$ with an empty co-basis has complexity $2^{n-1}$.
\item
If an atom of $R_n(a,b,c,d)$ has a co-basis of size $r$ with $1 \le r \le n-1$, its  complexity is 
\[ 1 + \sum_{k=1}^{r} \sum_{h=k+1}^{k+n-r} \binom{n-1}{h-1}\binom{h-1}{k}. \]
\item
The reverse of $R_n(a,d)$ has complexity $2^{n-1}$.
\item
For $n\ge2$, the star of $R_n(a,d)$ has complexity $n+1$.
\item
For $m,n \ge 3$, the complexity of $R_m(a,b,d) \cup R_n(b,a,d)$ is $mn-(m+n-2)$.
\item
For $m,n\ge 3$, the complexity of $R_m(a,b,d) \cap R_n(b,a,d)$ is $mn$.
\item
For $m,n\ge 3$, the complexity of $R_m(a,b,d) \setminus R_n(b,a,d)$ is $mn-(m-1)$.
\item
For $m,n\ge 3$, the complexity of $R_m(a,b,d) \oplus R_n(b,a,d)$ is $mn$.
\item
For $m,n\ge 3$, since any binary boolean operation can be expressed as a combination of the four operations above (and complement, which does not affect complexity), the complexity of $R_m(a,b,d) \circ R_n(b,a,d)$ is maximal for all binary boolean operations $\circ$.
\item
For $m,n\ge 3$, if $m \ne n$, then the complexity of $R_m(a,b,d) \circ R_n(a,b,d)$ is maximal for all binary boolean operations $\circ$.
\item
The complexity of $R_m(a,b,d)\cdot R_n(a,b,d)$ is $m+2^{n-2}$.
\ei
\end{theorem}

These claims are proved in the remainder of the paper. 

\section{Conditions for the Complexity of  Right Ideals}
\label{sec:cond}

We examine the conditions for the complexity of  a regular right ideal following the list introduced in~\cite{Brz13}.

\subsection{Properties of a Single Language}
\label{sec:1lan}

\noin {\bf A0 (Complexity of the Language):}
$R_n(a,d)$ has $n$ quotients because the DFA $\cR_n(a,d)$ is minimal. This holds since the non-final state $i$ accepts $a^{n-1-i}d$ and no other non-final state accepts this word, for $1\le i\le n-1$, and  all non-final states are distinguishable from the final state $n$. Hence no two states are equivalent.
\smallskip

\noin {\bf A1 (Cardinality of the Syntactic Semigroup):}
It was proved in~\cite{BrYe11} that the syntactic semigroup of a right ideal of complexity $n$ has cardinality at most $n^{n-1}$. To show $R_n(a,b,c,d)$ meets this bound, one first verifies the following:
\begin{remark}
\label{rem:transposition}
For $n\ge 3$,  the transposition $(1,2)$ in $\cR_n$ is induced by  $a^{n-2}b$.
\end{remark}

\begin{theorem}[Syntactic Semigroup]
The syntactic semigroup of the language $R_n(a,b,c,d)$ has cardinality $n^{n-1}$.
\end{theorem}
\begin{proof}
The cases $n\le 3$ are easily checked. For $n\ge 4$,
let the DFA $\cP_n$ be $\cP_n=(Q,\Sig,\delta, 1,\{n\})$, where $Q=\{1,\ldots,n\}$, $\Sig=\{a,b,c,d\}$, and
$a\colon(1,\ldots,n-1)$, $b\colon (1,2)$, ${c\colon(n-1 \rightarrow 1)}$ and 
${{d\colon(n-1\rightarrow n)}}$.
It was proved in~\cite{BrYe11} that  the syntactic semigroup of $P_n(a,b,c,d)$ has cardinality $n^{n-1}$. Since words in $\Sig^*$ can induce all the transformations of $\cP_n$ in $\cR_n(a,b,c,d)$,
the claim follows. \qed
\end{proof}
\smallskip

\noin {\bf A2 (Complexity of Quotients):} 
Each quotient of $R_n(a,d)$, except the quotient $\{a,d\}^*$, has complexity $n$, since the states $1,\ldots, n-1$ are strongly connected.
Hence the complexities of the quotients are as high as possible for right ideals.
\smallskip

\noin {\bf A3 (Number of Atoms):}
It was proved  in~\cite{BrTa12} that the number of atoms of $L$ is precisely the complexity of the reverse of $L$.
It was shown in~\cite{BJL13} that the maximal complexity of $L^R$ for right ideals is 
$2^{n-1}$.
For $n\le 3$ it is easily checked that our witness meets this bound.
For $n>3$, it was proved in~\cite{BrYe11} that the reverse of $R_n(a,d)$, and hence also of $R_n(a,b,c,d)$, reaches this bound.
\smallskip

\noin {\bf A4 (Complexity of Atoms):}
This is the topic of Section~\ref{sec:atoms}.

\subsection{Unary Operations}
\label{sec:unary}
\noin {\bf B1 (Reversal):}
See {\bf A3}.
\smallskip

\noin {\bf B2 (Star):}
It was proved in~\cite{BJL13} that the complexity of the star of a right ideal of complexity $n$ is at most $n+1$.

\begin{theorem}[Star]
\label{thm:star}
For $n\ge 2$, the complexity of $(R_n(a,d))^*$ is $n+1$.
\end{theorem}
\begin{proof}
The complexity of $R_1^*$ is 1.
For $n>1$, let $\cN_n$ be the $\eps$-NFA obtained by taking $\cR_n(a,d)$, adding a new  initial state $s$ which is also a final state, with the same transitions as state 1, and a transition from state $n$ to $1$ on $\varepsilon$. This NFA recognizes $(R_n(a,d))^*$. 
Let $N_n^{\deter}$ be the DFA obtained from $\cN_n$ by the subset construction, where only reachable states are used. 
We show the DFA $N_n^{\deter}$ has at least $n+1$ reachable and pairwise distinguishable states, and thus has exactly $n+1$ states, since $n+1$ is an upper bound.

Each state of $N_n^{\deter}$ is a subset of $Q \cup \{s\}$. The initial state is $\{s\}$. For $n=2$, we reach $\{1\}$ by $a$ and $\{2\}$ by $d$. For $n \ge 3$, from $\{s\}$ we reach $\{2\}$ by $a$, from $\{2\}$ we  reach $\{3\}$, $\{4\}$, \dots, $\{n-1\}$ and $\{1\}$ by words in $\{a\}^*$, and from $\{n-1\}$ we reach $\{n\}$ by $d$. Thus $n+1$ subsets are reachable.

Subset $\{s\}$ is distinguishable from $\{n\}$, since $a$ is not accepted from $\{s\}$. Since $\{s\}$ and $\{n\}$ are the only final states, they are distinguishable from all other states.
For $1 \le i,j \le n-1$, $i \ne j$, $\{i\}$ is distinguishable from $\{j\}$ since $a^{n-1-i}d$ is accepted from $\{i\}$ but not $\{j\}$. 
Thus the $n+1$ reachable subsets are pairwise distinguishable. It follows that $(R_n(a,d))^*$ has complexity $n+1$. \qed
\end{proof}

\subsection{Binary Operations}

\noin{\bf C1 (Boolean Operations):} 
See Section~\ref{sec:boolean}.
\smallskip

\noin{\bf C2 (Product):}
See Section~\ref{sec:product}.

\section{Complexity  of Atoms}
\label{sec:atoms}
In~\cite{BrTa12}, for the language stream $(L_n \mid n \ge 2)$ described after Definition~\ref{def:mostcomplex}, it
was proved that the atoms of $L_n$ have maximal complexity amongst all regular languages of complexity $n$. Our goal in this section is to prove that the atoms of $R_n(a,b,c,d)$ have maximal complexity amongst all regular right ideals of complexity $n$. We follow the same approach as~\cite{BrTa12}:
\be
\item
Derive upper bounds for the complexities of atoms in the case of right ideals.
\item
Describe the transition function of the \'atomaton of $R_n(a,b,c,d)$.
\item
Prove that certain strong-connectedness and reachability results hold for states of minimal DFAs of atoms of $R_n(a,b,c,d)$.
\item
Using these results, prove that the complexity of each atom of $R_n(a,b,c,d)$ meets the established bound.
\ee
In fact, many steps of the following proof are similar or identical to the proof for $L_n$ given in~\cite{BrTa12}. Rather than reproducing all the arguments in full detail, we refer to this paper when appropriate.

\subsection{Upper Bounds}
Observe that the co-basis of an atom cannot contain $\Sig^*$; if it did, then $\ol{\Sig^*} = \emp$ would be a term in the corresponding atomic intersection and so the intersection would be empty. Thus, since all right ideals have $\Sig^*$ as a quotient, every atom of a right ideal must contain $\Sig^*$ in its basis, rather than in its co-basis. It follows the co-basis of an atom of a right ideal is either empty or contains $r$ quotients, where $1 \le r \le n-1$. Bounds for complexity in each case are now given.

\begin{proposition}[Complexity Bounds for Atoms of Right Ideals]
\label{prop:atomqcomp}
Let $n\ge 1$, let $L$ be a right ideal with  complexity $n$ and let $A$ be an atom of $L$. 
\be
\item
If $\ol{\cB}(A) = \emp$, the complexity of $A$ is at most $2^{n-1}$.
\item
If $|\ol{\cB}(A)| = r$ for $1 \le r \le n-1$, the  complexity of $A$ is at most
\begin{equation}
\label{eq:bound}
 1 + \sum_{k=1}^{r} \sum_{h=k+1}^{k+n-r} \binom{n-1}{h-1}\binom{h-1}{k}. 
\end{equation}
\ee
\end{proposition}

These bounds can be derived using the same counting arguments as in~\cite{BrTa12}.
To summarize briefly, suppose $A$ is an atom and consider a quotient $w^{-1}A$. Since the quotient operation distributes over intersection, $w^{-1}A$ is an intersection of uncomplemented quotients from the set $w^{-1}(\cB(A)) = \{ w^{-1}K_i \mid K_i \in \cB(A)\}$ and complemented quotients from $w^{-1}(\ol{\cB}(A)) = \{w^{-1}K_i \mid K_i \in \ol{\cB}(A)\}$. 
In Equation (\ref{eq:bound}), $k$ represents the possible sizes of $w^{-1}(\ol{\cB}(A))$, while $h$ represents the possible sizes of $w^{-1}(\cB(A)) \cup w^{-1}(\ol{\cB}(A))$. The stated bounds follow from the observation that $\Sig^*$ must occur in $w^{-1}(\cB(A))$ but cannot occur in $w^{-1}(\ol{\cB}(A))$.

When $n \le 3$, the atoms of $R_n$ meet the bounds stated above:
\be
\item
For $R_1=\Sig^*$, there is only one atom, $\Sig^*$. It has an empty co-basis and meets the bound $2^{1-1}=1$. 
\item
For  $R_2=\{a,b,c\}^*d\{a,b,c,d\}^*$, the quotients are $K_1=\{a,b,c,d\}^*$ and 
$K_2=R_2$.
The atom $K_1\cap K_2=K_2$ has an empty co-basis and meets the bound $2^{2-1}=2$.
The only other atom is $K_1\cap\ol{K_2}$ and it meets the bound of Equation~(\ref{eq:bound}), which is also 2.
\item
For $n=3$, one verifies that input $b$ can be omitted, and that
the four atoms of $R_3(a,c,d)$ meet the required bounds.
\ee
Henceforth we assume that $n \ge 4$.

\subsection{Structure of the \'Atomaton}
The notion of an \emph{interval} will be useful in the following sections. If $U$ and $V$ are sets, the \emph{interval $[V,U]$ between $V$ and $U$} is the set of all subsets of $U$ that contain $V$. Intervals are sets of sets, but we often refer to them as \emph{collections of sets} to reduce confusion. If $V$ is not a subset of $U$, then $[V,U]$ is empty.

Let $\cA$ denote the \'atomaton of $R_n$. By Proposition \ref{prop:atomiso}, $\cA$ is isomorphic to $\cR^{\rev\deter\rev}_n$ by the map $\varphi$, and so we can treat the states of $\cA$ as subsets of the state set $Q$ of $\cR_n$. Under $\varphi$, the quotient $\Sig^*$ of $R_n$ corresponds to state $n$ of $\cR_n$. Hence, since all atoms of $R_n$ contain $\Sig^*$ in their basis, all states of $\cA$ are subsets of $Q$ that contain $n$. 
There are $2^{n-1}$ such subsets, and $\cA$ has $2^{n-1}$ states (since $R_n$ has $2^{n-1}$ atoms). Hence the set of states of $\cA$ is the interval $[\{n\},Q]$, that is, the collection of all subsets of $Q$ containing $n$. 
The initial atoms of $\cA$ are those that contain $L$ in their basis; under $\varphi$ these become subsets that contain state $1$, and thus the set of initial states is the interval $[\{1,n\},Q]$.
The final atom of $\cA$ is the atom whose basis contains all the quotients of $R_n$ that contain $\eps$, and no other quotients. The only quotient containing $\eps$ is $\Sig^*$, so under $\varphi$ the final atom becomes the subset $\{n\}$.

We have now established the set of states of the \'atomaton as well as the sets of initial and final states. Next we describe the transition function.
\begin{proposition}[States and Transitions of the \'Atomaton]
The \'atomaton of $R_n$ is $\cA = ([\{n\},Q],\Sig,\eta,[\{1,n\},Q],\{\{n\}\})$. If $S \in [\{n\},Q]$, then:
\be
\item
$\eta(S,a) = \displaystyle{\bigcup_{q \in S} \delta(q,a)}$, where $\delta$ is the transition function of $\cR_n$,
\item
$\eta(S,b) = \displaystyle{\bigcup_{q \in S} \delta(q,b)}$,
\item
If $S \cap \{1,n-1\} = \emp$, then:
    \be
    \item $\eta(S,c) = \{S,S \cup \{n-1\}\}$,
    \item $\eta(S \cup \{n-1\},c) = \emp$,
    \item $\eta(S \cup \{1\},c) = \emp$, and
    \item $\eta(S \cup \{1,n-1\},c) = \{S \cup \{1\}, S \cup \{1,n-1\}\}$.
    \ee
\item
If $S \cap \{n-1\} = \emp$, then:
    \be
    \item $\eta(S,d) = \emp$, and 
    \item $\eta(S \cup \{n-1\},d) = \{S, S \cup \{n-1\}\}$.
    \ee
\ee
\end{proposition}

The transition function of $\cA$ can be derived using the same method as in~\cite{BrTa12}; we take the transition function of $\cR_n$ and track how it changes as the automaton is reversed, determinized, then reversed again. 

\subsection{Strong-Connectedness and Reachability}
We now consider minimal DFAs of atoms of $R_n$. If $S \subseteq Q$ is a state of $\cA$, then its language is some atom $A_S$ of $R_n$.
If we take $\cA$ and change the set of initial states to $\{S\}$, we obtain an NFA $\cA_S = ([\{n\},Q],\Sig,\eta,\{S\},\{\{n\}\})$ that recognizes $A_S$. 
It was proved in~\cite{Brz63} that, if $\cN$ is any NFA that has no empty states and is such that $\cN^R$ is deterministic, then $\cN^D$ is minimal.
Since $\cA$ has no empty state and $\cA^R$ is deterministic, so is $\cA^R_S$; hence 
$\cA^{\deter}_S$ is the minimal DFA of $A_S$. 
Therefore we can determine the complexity of the atom $A_S$ of $R_n$ by constructing the NFA $\cA_S$, determinizing, and then counting the number of states of $\cA^{\deter}_S$.

The states of $\cA_S$ are subsets of $Q$ lying in the interval $[\{n\},Q]$; thus the states of $\cA^{\deter}_S$ are collections of subsets of $Q$ from this interval. In fact, they are not just arbitrary collections; we will see that every state of $\cA^{\deter}_S$ is a (possibly empty) subinterval of $[\{n\},Q]$. For conciseness, we refer to these subintervals of $[\{n\},Q]$ as $R_n$-intervals. The initial state of $\cA^{\deter}_S$ is the $R_n$-interval $[S,S] = \{S\}$. The next two results allow us to determine which other $R_n$-intervals are reachable from this initial state.

We assign a \emph{type} to each non-empty interval as follows: the type of $[V,U]$ is the ordered pair $(v,u)$, where $|V| = v$ and $|U| = u$. The empty interval has no type.

\begin{lemma}[Strong-Connectedness of Intervals]
\label{lem:connect}
If $A_S$ is an atom of $R_n$, then in its minimal DFA $\cA^{\deter}_S$, all $R_n$-intervals which have the same type are strongly connected by words in $\{a,b\}^*$.
\end{lemma}
\begin{proof}
Since $a$ induces the cycle $(1,\dotsc,n-1)$ and $a^{n-2}b$ induces the transposition $(1,2)$, we see that by Proposition \ref{prop:piccard2} and Remark \ref{rem:transposition}, words in $\{a,b\}^*$ can induce any permutation of $Q$ that fixes $\{n\}$. 

Let $[V_1,U_1]$ and $[V_2,U_2]$ be $R_n$-intervals of the same type. We can assume they are non-empty, and thus $V_1 \subseteq U_1$ and $V_2 \subseteq U_2$. Since these intervals have the same type, it follows $|V_1| = |V_2|$ and $|U_1 \setminus V_1| = |U_2 \setminus V_2|$. Hence there exists a bijection $\pi\colon Q \to Q$ that maps $V_1$ onto $V_2$ and $U_1 \setminus V_1$ onto $U_2 \setminus V_2$. Furthermore, since these are $R_n$-intervals, we have $n \in V_1 \cap V_2$. Thus without loss of generality we can assume $\pi$ fixes $n$ and maps $V_1 \setminus \{n\}$ onto $V_2 \setminus \{n\}$. Since the bijection $\pi$ is a permutation of $Q$ that fixes $\{n\}$, it can be induced by words in $\{a,b\}^*$, and the result follows. \qed
\end{proof}

\begin{lemma}[Reachability]
\label{lem:reach}
If $A_S$ is an atom of $R_n$, then in its minimal DFA $\cA^{\deter}_S$, the following holds:
From an $R_n$-interval of type $(v,u)$, 
if $v \ge 2$ we can reach an $R_n$-interval of type $(v-1,u)$, 
and if $u \le n-2$ we can reach an $R_n$-interval of type $(v,u+1)$.
\end{lemma}
\begin{proof}
Let $[V,U]$ be an $R_n$-interval of type $(v,u)$. If $v \ge 2$, then by Lemma \ref{lem:connect} we can reach an $R_n$-interval $[V',U']$ of type $(v,u)$ such that $n-1 \in V$. Then by  $d$ we can reach the $R_n$-interval $[V' \setminus \{n-1\},U']$ of type $(v-1,u)$. Note that $v \ge 2$ is required: if $v = 1$, then $V = \{n\}$ since $R_n$-intervals must contain $n$.

If $u \le n-2$, then again by Lemma \ref{lem:connect} we can reach an $R_n$-interval $[V',U']'$ of type $(v,u)$ such that $U \cap \{1,n-1\} = \emp$. Then by input $c$ we can reach the $R_n$-interval $[V',U' \cup \{n-1\}]$ of type $(v,u+1)$. \qed
\end{proof}

\subsection{Counting Reachable Intervals}
For each atom $A_S$ of $R_n$, we count the number of reachable $R_n$-intervals in the minimal DFA $\cA^{\deter}_S$. 
We will see that, for each atom, the number of reachable $R_n$-intervals in the minimal DFA matches the upper bounds on complexity stated in Proposition \ref{prop:atomqcomp}. This shows that every state of $\cA^{\deter}_S$ is an $R_n$-interval, as we claimed earlier, and proves that $A_S$ has maximal complexity.

\begin{theorem}[Atoms, Empty Co-Basis]
For $n \ge 1$, the atom of $R_n$ with an empty co-basis has complexity $2^{n-1}$.
\end{theorem}
\begin{proof}
Since the cases for $n<4$ have already been handled, assume that $n\ge 4$.
The atom with an empty co-basis is $A_Q$. Consider $\cA^{\deter}_Q$, the minimal DFA of this atom. The initial state of this DFA is the $R_n$-interval $[Q,Q]$ of type $(n,n)$.

By Lemma \ref{lem:reach}, we can reach $R_n$-intervals of types $(n-1,n)$, $(n-2,n)$, \dots, $(1,n)$. By Lemma \ref{lem:connect} we can reach all $R_n$-intervals of these types. There are $\binom{n-1}{k-1}$ $R_n$-intervals of type $(k,n)$, since if $[V,U]$ is an $R_n$-interval then $V$ must contain $n$ and the remaining $k-1$ elements are chosen arbitrarily from $U \setminus \{n\}$. Thus the total number of reachable $R_n$-intervals is at least
\[ \sum_{k=1}^{n} \binom{n-1}{k-1} = 2^{n-1}. \]
Thus $A_Q$ has at least $2^{n-1}$ quotients. By Proposition \ref{prop:atomqcomp}, $2^{n-1}$ is an upper bound on the number of quotients of $A_Q$, and thus $A_Q$ has exactly $2^{n-1}$ quotients. \qed
\end{proof}

Recall that there are no atoms of $R_n$ with a co-basis of size $n$, since each atom has $\Sig^*$ in its basis. We consider atoms with between $1$ and $n-1$ quotients in their co-basis.

\begin{theorem}[Atoms, Non-Empty Co-Basis]
For $n \ge 2$, each atom of $R_n$ with a co-basis of size $r$, where $1 \le r \le n-1$, has quotient complexity
\[ f(n,r) = 1 + \sum_{k=1}^{r} \sum_{h=k+1}^{k+n-r} \binom{n-1}{h-1}\binom{h-1}{k}. \]
\end{theorem}

\begin{proof}
Since the calculations  here are nearly identical to those of~\cite{BrTa12},  we omit most of the details. 
Let $A_S$ be an atom, where $S$ is a proper subset of $Q$ that contains $n$. The minimal DFA of $A_S$ is $\cA^{\deter}_S$, and its initial state is the $R_n$-interval $[S,S]$ of type $(n-r,n-r)$. By Lemmas \ref{lem:connect} and \ref{lem:reach} and a counting argument, one can show that the number of non-empty reachable $R_n$-intervals is at least
\[ \sum_{u=n-r}^{n-1}\sum_{v=1}^{n-r}\binom{n-1}{u-1}\binom{u-1}{v-1}. \]
If $[V,U]$ is a non-empty reachable $R_n$-interval, $v=|V|$ and $u=|U|$ are the possible sizes of $V$ and $U$. Algebraic manipulation shows that the bound above is equal to
\[ \sum_{k=1}^{r}\sum_{h=k+1}^{k+n-r}\binom{n-1}{h-1}\binom{h-1}{k}. \]
That is, the number of non-empty reachable $R_n$-intervals is at least $f(n,r)-1$. The empty interval is also reachable (for example, by input $d$ from the interval $[\{n\},Q \setminus\{n-1\}]$) and thus the number of reachable intervals is at least $f(n,r)$. Since $f(n,r)$ is an upper bound by Proposition \ref{prop:atomqcomp}, the result follows. \qed
\end{proof}

Table~\ref{tab:atomcomp} shows the bounds for right ideals (first entry) and compares them to those of regular languages (second entry). An asterisk indicates the case is impossible for right ideals. The \emph{ratio} row shows the ratio $m_n/m_{n-1}$ for $n \ge 2$, where $m_i$ is the $i^{\rm{th}}$ entry in the \emph{max} row.
The entries that are maximal for a given $n$ are shown in boldface type.

\begin{table}[ht]
\caption{Maximal quotient complexity of atoms of right ideals.}
\label{tab:atomcomp}
\begin{center}
$
\begin{array}{| c| c|c| c|c| c|c| c|c|}    
\hline
\ \ n \ \ 
&\ \ 1 \ \ &\ \ 2 \ \ &  \ 3 \ &\ 4 \ &\ 5 \ &  \ 6 \ &\ 7 \ &\cdots  
\\
\hline
\hline  
$r=0$
&  \bf 1/1 &  \bf 2/3 & \ 4/7  \ & \ 8/15 \  & \ 16/31 \ & \ 32/63  \  & \  64/127 \ & \cdots\\
\hline  
$r=1$
&  \bf \ast/1 & \bf 2/3  & \ \bf 5/10  \ & \ 13/29  \  & \ 33/76 \ & \  81/187  \  & \  193/442  \ & \cdots\\
\hline  
$r=2$ 
&  &  \bf \ast /3 & \ 4\bf /10 \ & \ \bf 16/43  \  & \ \bf 53/141  \ & \ 156/406   \  & \ 427/1,086  \ & \cdots\\
\hline  
$r=3$
&  &  & \ \ast/7 \ & \  8/29 \  & \ 43/\bf 141 \ & \ \bf 166/501  \  & \ \bf  542/1,548  \ & \cdots\\
\hline  
$r=4$
&   &  & \  \ & \  \ast/15 \  & \ 16/76 \ & \ 106/406  \  & \ 462/\bf 1,548   \ & \cdots\\
\hline  
$r=5$
&   &  & \   \ & \ \  & \ \ast/31 \ & \ 32/187  \  & \  249/1,086   \ & \ \cdots\\
\hline
$r=6$
&   &  & \   \ & \ \  & \ \ & \ \ast/63  \  & \  64/442   \ & \ \cdots\\
\hline
\ max \
& 1/1 & 2/3  & \ 5/10 \ & \ 16/43  \  & \ 53/141 \ & \ 166/501\  & \ 542/1,548 \ &\cdots \\
\hline  
ratio
& - & 2/3  & \ 2.50/3.33 \ & \ 3.20/4.30  \  & \ 3.31/3.28 \ & \ 3.13/3.55 \  & \ 3.27/3.09 \ & \cdots \\
\hline
\end{array}
$
\end{center}
\end{table}

\section{Boolean Operations}
\label{sec:boolean}
Since $K_n\cup K_n=K_n\cap K_n=K_n$, and 
$K_n\setminus K_n=K_n \oplus K_n=\emp$, two different languages have to be used
to reach the bounds for boolean operations if $m=n$. 
Figure~\ref{fig:binary} shows the  DFAs $\cR_4(a,b,d)$ and  $\cR_5(b,a,d)$. 
The 
direct product  of $\cR_4(a,b,d)$ and  $\cR_5(b,a,d)$ is  in
Figure~\ref{fig:cross}, where transitions under $a$ and $d$ are shown with solid lines and under $b$, with dotted lines. Self-loops are omitted.

In general, let $\cR_m = \cR_m(a,b,d)$,  $\cR_n = \cR_n(b,a,d)$,
and $\cR_{m,n} = \cR_m \times \cR_n = (Q_m \times Q_n, \Sig, \delta, (1,1), F)$ with $\delta((i,j),\sig) = (\delta_m(i,\sig),\delta_n(j,\sig))$, where $\delta_m$ ($\delta_n$) is the transition function of $\cR_m$ ($\cR_n$). Depending on $F$, this DFA recognizes different boolean operations with $R_m$ and $R_n$.

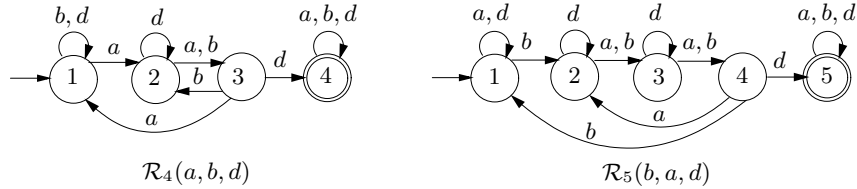
\begin{figure}[hbt]
\begin{center}
\setlength{\unitlength}{0.00039370in}
\begingroup\makeatletter\ifx\SetFigFont\undefined%
\gdef\SetFigFont#1#2#3#4#5{%
  \reset@font\fontsize{#1}{#2pt}%
  \fontfamily{#3}\fontseries{#4}\fontshape{#5}%
  \selectfont}%
\fi\endgroup%
{\renewcommand{\dashlinestretch}{30}
\begin{picture}(11110,2486)(0,-10)
\put(10399,2216){\makebox(0,0)[lb]{\smash{{\SetFigFont{9}{10.8}{\familydefault}{\mddefault}{\updefault}$a,b,d$}}}}
\put(8554.500,1870.929){\arc{394.717}{2.4948}{6.9299}}
\blacken\thicklines
\path(8714.033,1892.097)(8712.000,1752.000)(8785.998,1870.977)(8714.033,1892.097)
\thinlines
\put(10797.500,1870.929){\arc{394.717}{2.4948}{6.9299}}
\blacken\thicklines
\path(10957.033,1892.097)(10955.000,1752.000)(11028.998,1870.977)(10957.033,1892.097)
\thinlines
\put(851.500,1862.929){\arc{394.717}{2.4948}{6.9299}}
\blacken\thicklines
\path(1011.033,1884.097)(1009.000,1744.000)(1082.998,1862.977)(1011.033,1884.097)
\thinlines
\put(1931.500,1847.929){\arc{394.717}{2.4948}{6.9299}}
\blacken\thicklines
\path(2091.033,1869.097)(2089.000,1729.000)(2162.998,1847.977)(2091.033,1869.097)
\thinlines
\put(7466.500,1885.929){\arc{394.717}{2.4948}{6.9299}}
\blacken\thicklines
\path(7626.033,1907.097)(7624.000,1767.000)(7697.998,1885.977)(7626.033,1907.097)
\thinlines
\put(6401.500,1870.929){\arc{394.717}{2.4948}{6.9299}}
\blacken\thicklines
\path(6561.033,1892.097)(6559.000,1752.000)(6632.998,1870.977)(6561.033,1892.097)
\thinlines
\put(4194,1454){\ellipse{630}{630}}
\put(3065,1466){\ellipse{630}{630}}
\put(8552,1460){\ellipse{630}{630}}
\put(9673,1467){\ellipse{630}{630}}
\put(6402,1452){\ellipse{630}{630}}
\put(7459,1470){\ellipse{630}{630}}
\put(10787,1459){\ellipse{630}{630}}
\put(10788,1457){\ellipse{540}{540}}
\put(847,1435){\ellipse{630}{630}}
\put(1930,1430){\ellipse{630}{630}}
\put(4195,1451){\ellipse{540}{540}}
\path(2172,1654)(2802,1654)
\blacken\thicklines
\path(2667.000,1616.500)(2802.000,1654.000)(2667.000,1691.500)(2667.000,1616.500)
\thinlines
\path(1047,1661)(1677,1661)
\blacken\thicklines
\path(1542.000,1623.500)(1677.000,1661.000)(1542.000,1698.500)(1542.000,1623.500)
\thinlines
\path(7721,1685)(8306,1685)
\blacken\thicklines
\path(8171.000,1647.500)(8306.000,1685.000)(8171.000,1722.500)(8171.000,1647.500)
\thinlines
\path(8817,1677)(9402,1677)
\blacken\thicklines
\path(9267.000,1639.500)(9402.000,1677.000)(9267.000,1714.500)(9267.000,1639.500)
\thinlines
\path(6627,1692)(7212,1692)
\blacken\thicklines
\path(7077.000,1654.500)(7212.000,1692.000)(7077.000,1729.500)(7077.000,1654.500)
\thinlines
\path(5577,1459)(6064,1459)
\blacken\thicklines
\path(5944.000,1429.000)(6064.000,1459.000)(5944.000,1489.000)(5944.000,1429.000)
\thinlines
\path(12,1452)(499,1452)
\blacken\thicklines
\path(379.000,1422.000)(499.000,1452.000)(379.000,1482.000)(379.000,1422.000)
\thinlines
\path(3372,1459)(3867,1459)
\blacken\thicklines
\path(3732.000,1421.500)(3867.000,1459.000)(3732.000,1496.500)(3732.000,1421.500)
\thinlines
\path(2810,1279)(2225,1279)
\blacken\thicklines
\path(2360.000,1316.500)(2225.000,1279.000)(2360.000,1241.500)(2360.000,1316.500)
\thinlines
\path(10002,1467)(10497,1467)
\blacken\thicklines
\path(10362.000,1429.500)(10497.000,1467.000)(10362.000,1504.500)(10362.000,1429.500)
\thinlines
\path(2922,1189)(2921,1188)(2919,1186)
	(2915,1183)(2909,1178)(2900,1171)
	(2888,1162)(2874,1150)(2857,1136)
	(2836,1120)(2813,1102)(2788,1083)
	(2760,1063)(2730,1041)(2699,1019)
	(2666,996)(2632,974)(2596,951)
	(2559,929)(2521,907)(2482,886)
	(2442,866)(2400,846)(2356,828)
	(2310,811)(2262,795)(2213,780)
	(2160,767)(2106,757)(2049,748)
	(1991,742)(1932,739)(1869,739)
	(1809,743)(1750,750)(1695,759)
	(1643,771)(1594,785)(1547,800)
	(1503,817)(1462,835)(1422,854)
	(1384,875)(1348,896)(1314,918)
	(1280,940)(1249,963)(1218,985)
	(1190,1007)(1163,1029)(1139,1049)
	(1116,1068)(1097,1085)(1079,1100)
	(1065,1113)(1054,1124)(1045,1132)(1032,1144)
\blacken\thicklines
\path(1156.634,1079.987)(1032.000,1144.000)(1105.763,1024.877)(1156.634,1079.987)
\thinlines
\path(9492,1196)(9491,1195)(9489,1193)
	(9485,1190)(9478,1184)(9469,1177)
	(9457,1167)(9442,1155)(9425,1141)
	(9404,1125)(9381,1107)(9356,1088)
	(9328,1068)(9299,1047)(9268,1026)
	(9235,1005)(9201,983)(9166,963)
	(9130,943)(9091,923)(9052,905)
	(9010,887)(8967,871)(8921,856)
	(8872,842)(8821,831)(8767,821)
	(8711,813)(8652,809)(8592,807)
	(8532,809)(8473,813)(8417,821)
	(8362,831)(8311,842)(8262,856)
	(8216,871)(8172,887)(8130,905)
	(8090,923)(8052,943)(8015,963)
	(7979,983)(7945,1005)(7912,1026)
	(7880,1047)(7850,1068)(7823,1088)
	(7797,1107)(7773,1125)(7752,1141)
	(7734,1155)(7719,1167)(7707,1177)
	(7698,1184)(7684,1196)
\blacken\thicklines
\path(7810.904,1136.615)(7684.000,1196.000)(7762.095,1079.671)(7810.904,1136.615)
\thinlines
\path(9717,1159)(9716,1158)(9715,1157)
	(9711,1155)(9706,1151)(9699,1145)
	(9689,1137)(9676,1128)(9660,1116)
	(9642,1102)(9620,1086)(9596,1069)
	(9569,1049)(9539,1028)(9507,1005)
	(9472,982)(9436,957)(9397,932)
	(9357,906)(9316,880)(9272,854)
	(9228,828)(9182,802)(9135,777)
	(9087,752)(9038,728)(8987,705)
	(8934,682)(8880,661)(8824,641)
	(8766,622)(8706,604)(8643,587)
	(8578,573)(8511,560)(8441,549)
	(8369,540)(8294,533)(8219,530)
	(8142,529)(8066,532)(7990,537)
	(7917,545)(7846,556)(7777,568)
	(7712,583)(7649,599)(7589,617)
	(7531,637)(7475,657)(7422,679)
	(7370,701)(7321,725)(7273,749)
	(7227,775)(7181,800)(7138,827)
	(7095,854)(7054,881)(7015,908)
	(6976,935)(6939,961)(6904,988)
	(6871,1013)(6839,1038)(6810,1061)
	(6783,1083)(6758,1103)(6736,1121)
	(6717,1137)(6700,1152)(6686,1164)
	(6675,1174)(6666,1181)(6659,1187)(6649,1196)
\blacken\thicklines
\path(6774.431,1133.563)(6649.000,1196.000)(6724.259,1077.816)(6774.431,1133.563)
\put(6319,1380){\makebox(0,0)[lb]{\smash{{\SetFigFont{9}{10.8}{\rmdefault}{\mddefault}{\updefault}$1$}}}}
\put(7377,1387){\makebox(0,0)[lb]{\smash{{\SetFigFont{9}{10.8}{\rmdefault}{\mddefault}{\updefault}$2$}}}}
\put(8486,1380){\makebox(0,0)[lb]{\smash{{\SetFigFont{9}{10.8}{\rmdefault}{\mddefault}{\updefault}$3$}}}}
\put(9589,1388){\makebox(0,0)[lb]{\smash{{\SetFigFont{9}{10.8}{\rmdefault}{\mddefault}{\updefault}$4$}}}}
\put(10706,1388){\makebox(0,0)[lb]{\smash{{\SetFigFont{9}{10.8}{\rmdefault}{\mddefault}{\updefault}$5$}}}}
\put(1835,1361){\makebox(0,0)[lb]{\smash{{\SetFigFont{9}{10.8}{\rmdefault}{\mddefault}{\updefault}$2$}}}}
\put(4099,1382){\makebox(0,0)[lb]{\smash{{\SetFigFont{9}{10.8}{\rmdefault}{\mddefault}{\updefault}$4$}}}}
\put(754,1369){\makebox(0,0)[lb]{\smash{{\SetFigFont{9}{10.8}{\rmdefault}{\mddefault}{\updefault}$1$}}}}
\put(3822,2201){\makebox(0,0)[lb]{\smash{{\SetFigFont{9}{10.8}{\familydefault}{\mddefault}{\updefault}$a,b,d$}}}}
\put(2300,1789){\makebox(0,0)[lb]{\smash{{\SetFigFont{9}{10.8}{\familydefault}{\mddefault}{\updefault}$a,b$}}}}
\put(1790,829){\makebox(0,0)[lb]{\smash{{\SetFigFont{9}{10.8}{\familydefault}{\mddefault}{\updefault}$a$}}}}
\put(596,2179){\makebox(0,0)[lb]{\smash{{\SetFigFont{9}{10.8}{\familydefault}{\mddefault}{\updefault}$b,d$}}}}
\put(1865,2156){\makebox(0,0)[lb]{\smash{{\SetFigFont{9}{10.8}{\familydefault}{\mddefault}{\updefault}$d$}}}}
\put(1295,1774){\makebox(0,0)[lb]{\smash{{\SetFigFont{9}{10.8}{\familydefault}{\mddefault}{\updefault}$a$}}}}
\put(3485,1602){\makebox(0,0)[lb]{\smash{{\SetFigFont{9}{10.8}{\familydefault}{\mddefault}{\updefault}$d$}}}}
\put(2937,1368){\makebox(0,0)[lb]{\smash{{\SetFigFont{9}{10.8}{\rmdefault}{\mddefault}{\updefault}$3$}}}}
\put(2420,1369){\makebox(0,0)[lb]{\smash{{\SetFigFont{9}{10.8}{\familydefault}{\mddefault}{\updefault}$b$}}}}
\put(6754,1827){\makebox(0,0)[lb]{\smash{{\SetFigFont{9}{10.8}{\familydefault}{\mddefault}{\updefault}$b$}}}}
\put(10077,1587){\makebox(0,0)[lb]{\smash{{\SetFigFont{9}{10.8}{\familydefault}{\mddefault}{\updefault}$d$}}}}
\put(7789,1835){\makebox(0,0)[lb]{\smash{{\SetFigFont{9}{10.8}{\familydefault}{\mddefault}{\updefault}$a,b$}}}}
\put(8878,1835){\makebox(0,0)[lb]{\smash{{\SetFigFont{9}{10.8}{\familydefault}{\mddefault}{\updefault}$a,b$}}}}
\put(8488,881){\makebox(0,0)[lb]{\smash{{\SetFigFont{9}{10.8}{\familydefault}{\mddefault}{\updefault}$a$}}}}
\put(7819,109){\makebox(0,0)[lb]{\smash{{\SetFigFont{9}{10.8}{\rmdefault}{\mddefault}{\updefault}$\cR_5(b,a,d)$}}}}
\put(1752,117){\makebox(0,0)[lb]{\smash{{\SetFigFont{9}{10.8}{\rmdefault}{\mddefault}{\updefault}$\cR_4(a,b,d)$}}}}
\put(7602,650){\makebox(0,0)[lb]{\smash{{\SetFigFont{9}{10.8}{\familydefault}{\mddefault}{\updefault}$b$}}}}
\put(7355,2208){\makebox(0,0)[lb]{\smash{{\SetFigFont{9}{10.8}{\familydefault}{\mddefault}{\updefault}$d$}}}}
\put(8436,2224){\makebox(0,0)[lb]{\smash{{\SetFigFont{9}{10.8}{\familydefault}{\mddefault}{\updefault}$d$}}}}
\put(6108,2216){\makebox(0,0)[lb]{\smash{{\SetFigFont{9}{10.8}{\familydefault}{\mddefault}{\updefault}$a,d$}}}}
\thinlines
\put(4211.500,1870.929){\arc{394.717}{2.4948}{6.9299}}
\blacken\thicklines
\path(4371.033,1892.097)(4369.000,1752.000)(4442.998,1870.977)(4371.033,1892.097)
\end{picture}
}
\end{center}
\caption{Right-ideal witnesses for boolean operations.} 
\label{fig:binary}
\end{figure}

 \begin{figure}[h]
 \begin{center}
 \input cross.eepic
 \end{center}
 \caption{Cross-product automaton for boolean operations for $m=4,n=5$.}  
 \label{fig:cross}
 \end{figure}

In our proof that the bounds for boolean operations are reached, we use a result of Bell, Brzozowski, Moreira and Reis~\cite{BBMR13}.
We use the following terminology:
A binary boolean operation $\circ$ on regular languages is a  mapping $\circ:2^{\Sig^*}\times 2^{\Sig^*}\to 2^{\Sig^*}$.
If $L, L'\subseteq \Sig^*$, the result of the operation $\circ$ is denoted by $L\circ L'$.
We say that such a boolean operation is \emph{proper} if $\circ$ is not a constant ($\emp$ or $\Sig^*$)  and not a function of one variable only, that is, it is not the identity or complement function of one of the variables.

Let $S_n$ denote the symmetric group of degree $n$. 
A \emph{basis}~\cite{Pic39} of $S_n$
is an ordered pair $(s,t)$ of distinct transformations of $Q_n=\{1,\dots,n\}$ that generate $S_n$.
Two bases $(s,t)$ and $(s',t')$ of $S_n$ are \emph{conjugate} if there exists a transformation $r\in S_n$ such that $rsr^{-1}=s'$, and  $rtr^{-1}=t'$.
A DFA \emph{has a basis $(s,t)$ for $S_n$} if it has letters $a,b\in \Sig$ such that $a$ induces $s$ and $b$ induces $t$.

\begin{proposition}[Symmetric Groups and Boolean Operations~\cite{BBMR13}]
\label{prop:ssbool}
Suppose that $m,n,\ge1$,  $L_m$ and $L'_n$ are regular languages of complexity $m$ and $n$ respectively, and  $\cD_m$ and $\cD'_n$ are minimal DFAs for $L_m$ and $L'_n$ with $F$ and $F'$ as sets of final states. Suppose  that $D_m$  has a basis $B$ for $S_m$ and $D_n$  has a basis $B'$ for $S_n$.
 Let $\circ$ be a proper binary boolean function. Then the  following hold:
\be
\item
In the direct product $\cD_m \times \cD_n$, all $mn$ states are reachable if and only if $m \ne n$, or $m = n$ and the bases $B$ and $B'$ are not conjugate.
\item
For $m,n\ge 2$, but $(m,n)\not \in \{(2,2), (3,4),(4,3),(4,4)\}$, 
 the language $L_m\circ L_n$ has complexity $mn$ if and only if $m \ne n$, or $m = n$ and the bases $B$ and $B'$ are not conjugate.
\ee
\end{proposition}

Since the transition semigroup of $\cR_n$ has a basis for the symmetric group $S_{n-1}$, it contains all permutations of the set of non-final states $\{1,2,\dotsb,n-1\}$. This implies that,  in the direct product $\cR_{m,n}$, all states in the set $S = \{ (i,j) \mid 1 \le i \le m-1, 1 \le j \le n-1\}$ are reachable by words in $\{a,b\}^*$. 
Furthermore, 
if $m,n \ge 3$ and $(m,n) \not\in \{(3,3),(4,5),(5,4),(5,5)\}$, then every pair of states in $S$ is distinguishable with respect to $F\circ F'$. 
\newpage

\begin{theorem}[Boolean Operations]
\label{thm:bool2}
If $m,n\ge 3$, then
\be
\item
The complexity of $R_m(a,b,d) \cap R_n(b,a,d)$ is $mn$.
\item
The complexity of $R_m(a,b,d) \oplus R_n(b,a,d)$ is $mn$.
\item
The complexity of $R_m(a,b,d) \setminus R_n(b,a,d)$ is $mn-(m-1)$.
\item
The complexity of $R_m(a,b,d) \cup R_n(b,a,d)$ is $mn-(m+n-2)$.
\ee
\end{theorem}

\begin{proof}
In the cases where $(m,n) \in \{(3,3),(4,5),(5,4),(5,5)\}$, we cannot apply Proposition \ref{prop:ssbool}, but it is easy to verify computationally that the bounds are met. For the remainder of the proof we assume $(m,n) \not\in \{(3,3),(4,5),(5,4),(5,5)\}$.

Our first task is to show that all $mn$ states of $\cR_{m,n}$ are reachable.
By Proposition \ref{prop:ssbool}, all states in the set $S = \{(i,j) \mid 1 \le i \le m-1, 1 \le j \le n-1\}$ are reachable. The remaining states are the ones in the last row or last column (that is, row $m$ or column $n$) of the direct product.

For $1 \le j \le n-2$, from state $(m-1,j)$ we can reach $(m,j)$ by $d$. From state $(m,n-2)$ we can reach $(m,n-1)$ by $a$. From state $(m-1,n-1)$ we can reach $(m,n)$ by $d$. Hence all states in row $m$ are reachable.

For $1 \le i \le m-2$, from state $(i,n-1)$ we can reach $(i,n)$ by $d$. From state $(m-2,n)$ we can reach $(m-1,n)$ by $a$. Hence all states in column $n$ are reachable, and thus all $mn$ states are reachable.

We now count the number of distinguishable states for each operation. Let $H = \{(m,j) \mid 1 \le j \le n\}$ be the set of states in the last row and let $V = \{(i,n) \mid 1 \le i \le m\}$ be the set of states in the last column. If $\circ \in \{\cap,\oplus,\setminus,\cup\}$, then $R_m(a,b,d) \circ R_n(b,a,d)$ is recognized by $\cR_{m,n}$, where the set of final states is taken to be $H \circ V$.

By Proposition \ref{prop:ssbool}, all states of $\cR_{m,n}$ that lie in $S$ are distinguishable with respect to any non-empty strict subset of $S$. We claim that they are also distinguishable with respect to $H \circ V$ for $\circ \in \{\cap,\oplus,\setminus,\cup\}$.

To see this, let $H' = \{(m-1,j) \mid 1 \le j \le n-1\}$ and let $V' = \{(i,n-1) \mid 1 \le i \le m-1\}$. Then by Proposition \ref{prop:ssbool}, all states in $S$ are distinguishable with respect to $H' \cap V' = \{(m-1,n-1)\}$. This implies that for all pairs of states $(i,j),(k,\ell) \in S$, there exists a word $w$ that sends $(i,j)$ to $(m-1,n-1)$ and sends $(k,\ell)$ to some other state in $S$. It follows that the word $wd$ sends $(i,j)$ to $(m,n)$ (which is in $H \cap V$), while $(k,\ell)$ is sent to a state outside of $H \cap V$. Hence all states in $S$ are distinguishable with respect to $H \cap V$. The same argument works for $H \oplus V$, $H \setminus V$, and $H \cup V$.

Thus for each boolean operation $\circ$, all $(m-1)(n-1) = mn - m - n + 1$ states in $S$ are distinguishable with respect to the final state set $H \circ V$. To show that the complexity bounds are reached by $R_m(a,b,d) \circ R_n(b,a,d)$, it suffices to consider how many of the $m + n - 1$ states in $H \cup V$ are distinguishable with respect to $H \circ V$.

\noin\textbf{Intersection:}
Here the set of final states is $H \cap V = \{(m,n)\}$.
State $(m,n)$ is the only final state and hence is distinguishable from all the other states.
Any two states in $H$ ($V$) are distinguished by words in $b^*d$ ($a^*d$).
State $(m,1)$ accepts $b^{n-2}d$, while $(1,n)$ rejects it.
For $2\le i\le n-1$, $(m,i)$ is sent to $(m,1)$ by $b^{n-1-i}$, while 
state $(1,n)$ is not changed by that word.
Hence $(m,i)$ is distinguishable from $(1,n)$.
By a symmetric argument, $(j,n)$ is distinguishable from $(m,1)$ for $2 \le j \le m-1$.
For $2\le i\le n-1$ and $2\le j\le m-1$, $(m,i)$ is distinguished from
$(j,n)$ because $b^{n-i}$ sends the former to $(m,1)$ and the latter to a state of the form $(k,n)$, where $2\le k\le m-1$.
Hence all pairs of states from $H \cup V$ are distinguishable. There are $m + n - 1$ states in $H \cup V$, so it follows there are $(mn - m - n + 1) + (m + n - 1) = mn$ distinguishable states.

\noin\textbf{Symmetric Difference:}
Here the set of final states is $H \oplus V$, that is, all states in the last row and column except $(m,n)$, which is the only empty state.
This situation is complementary to that for intersection. Thus every two states from $H \cup V$ are distinguishable by the same word as for intersection.
Hence there are $mn$ distinguishable states.

\noin\textbf{Difference:}
Here the set of final states is $H \setminus V$, that is, all states in the last row $H$ except $(m,n)$, which is empty. All other states in the last column $V$ are also empty. 
The $m$ empty states in $V$ are all equivalent, and the $n-1$ final states in $H \setminus V$ are distinguished in the same way as for intersection. Hence there are $(n-1)+1 = n$ distinguishable states in $H \setminus V$. It follows there are $(mn - m - n + 1) + n = mn - (m-1)$ distinguishable states.

\noin\textbf{Union:}
Here the set of final states is $H \cup V$. From a state in $H \cup V$ it is only possible to reach other states in $H \cup V$, and all these states are final, so every state in $H \cup V$ accepts $\Sig^*$. Thus all the states in $H \cup V$ are equivalent, and thus there are $(mn - m - n + 1) + 1 = mn - (m + n - 2)$ distinguishable states.
\qed
\end{proof}

Although it is impossible for the stream $(R_n(a,b,d) \mid n\ge 3)$ to meet the bound for boolean operations when $m=n$, this stream is as complex as it could possibly be
in view of the following:
\begin{theorem}[Boolean Operations, $m\neq n$]
\noin
Suppose $m,n \ge 3$ and $m \ne n$.
\be
\item
The complexity of $R_m(a,b,d) \cup R_n(a,b,d)$ is $mn-(m+n-2)$.
\item
The complexity of $R_m(a,b,d) \cap R_n(a,b,d)$ is $mn$.
\item
The complexity of $R_m(a,b,d) \setminus R_n(a,b,d)$ is $mn-(m-1)$.
\item
The complexity of $R_m(a,b,d) \oplus R_n(a,b,d)$ is $mn$.
\ee
\end{theorem}
\begin{proof}
Let $\cR_m = \cR_m(a,b,d)$, $\cR_n = \cR_n(a,b,d)$, and $\cR_{m,n} = \cR_m \times \cR_n$ be the direct product automaton. If $(m,n) \in \{(4,5),(5,4)\}$, one can verify computationally that the bounds are met. If $(m,n) \not \in \{(4,5),(5,4)\}$,  we can apply Proposition \ref{prop:ssbool}. Thus by the arguments used in the proof of Theorem \ref{thm:bool2}, all states of $\cR_{m,n}$ are reachable. 
Furthermore, if $H = \{(m,j) \mid 1 \le j \le n\}$ and $V = \{(i,n) \mid 1 \le i \le m\}$, then all states in $S = \{(i,j) \mid 1 \le i \le m-1, 1 \le j \le n-1\}$ are distinguishable with respect to $H \circ V$ for each $\circ \in \{\cap,\oplus,\setminus,\cup\}$.
To determine the number of distinguishable states for each boolean operation $\circ$, it suffices to count the number of states in $H \cup V$ that are distinguishable with respect to $H \circ V$.

\noin\textbf{Intersection:}
Here the set of final states is $H \cap V = \{(m,n)\}$. Since $(m,n)$ is the only final state, it is distinguishable from all other states.
Any two states both in $H$ (or both in $V$) are distinguished by words in $a^*d$.
Suppose $m < n$. Then $a^{m-1}$ sends $(m,1)$ to $(m,m)$ and fixes $(1,n)$.
Words in $b^*$ can send $(m,m)$ to $(m,i)$ for $2 \le i \le n-1$, and they fix $(1,n)$.
For $2 \le i \le n-1$, $(m,i)$ accepts $b^{n-1-i}d$, while $(1,n)$ remains fixed.
Hence $(m,i)$ is distinguishable from $(1,n)$ for all $i$. 
For $2 \le i \le m-1$ and $2 \le j \le n-1$, $(m,i)$ is distinguished from $(j,n)$ because $a^{m-j}$ sends $(j,n)$ to $(1,n)$ and $(m,i)$ to some state that is distinguishable from $(1,n)$.
Hence all pairs of states from $H \cup V$ are distinguishable if $m < n$. A symmetric argument works for $m > n$. Thus all $mn$ states are distinguishable.

\noin\textbf{Symmetric Difference, Difference, and Union:}
The same arguments used in the proof of Theorem \ref{thm:bool2} work here. 
\qed
\end{proof}

\section{Product}
\label{sec:product}
We show that the complexity of the product of $R_m(a,b,d)$ with $R_n(a,b,d)$ 
reaches the maximum possible bound derived in~\cite{BJL13}.
To avoid confusing states of the two DFAs, we label their states differently.
Let $\cR_m=\cR_m(a,b,d)=(Q'_m, \Sig,\delta',q_1,\{ q_{m} \})$, where 
$Q'_m=\{q_1,\ldots,q_{m}\}$, and 
let $\cR_n=\cR_n(a,b,d)$, as in Definition~\ref{def:mostcomplex}.
Define the $\eps$-NFA $\cP=(Q'_m\cup Q_n,\Sig, \delta_\cP,\{q_1\},\{n\})$, 
where $\delta_\cP(q,a)=\{ \delta'(q,a)\}$ if $q\in Q'_m$, $a\in\Sig$,
$\delta_\cP(q,a)=\{ \delta(q,a)\}$ if $q\in Q_n$, $a\in \Sig$, 
and $\delta_\cP(q_{m},\eps)=\{1\}$.
This $\eps$-NFA accepts $R_mR_n$, and is illustrated in Figure~\ref{fig:product}.

\begin{figure}[hbt]
\begin{center}
\setlength{\unitlength}{0.00039370in}
\begingroup\makeatletter\ifx\SetFigFont\undefined%
\gdef\SetFigFont#1#2#3#4#5{%
  \reset@font\fontsize{#1}{#2pt}%
  \fontfamily{#3}\fontseries{#4}\fontshape{#5}%
  \selectfont}%
\fi\endgroup%
{\renewcommand{\dashlinestretch}{30}
\begin{picture}(11110,2486)(0,-10)
\put(2936,1391){\makebox(0,0)[lb]{\smash{{\SetFigFont{9}{10.8}{\rmdefault}{\mddefault}{\updefault}$q_3$}}}}
\put(8554.500,1870.929){\arc{394.717}{2.4948}{6.9299}}
\blacken\thicklines
\path(8714.033,1892.097)(8712.000,1752.000)(8785.998,1870.977)(8714.033,1892.097)
\thinlines
\put(10797.500,1870.929){\arc{394.717}{2.4948}{6.9299}}
\blacken\thicklines
\path(10957.033,1892.097)(10955.000,1752.000)(11028.998,1870.977)(10957.033,1892.097)
\thinlines
\put(851.500,1862.929){\arc{394.717}{2.4948}{6.9299}}
\blacken\thicklines
\path(1011.033,1884.097)(1009.000,1744.000)(1082.998,1862.977)(1011.033,1884.097)
\thinlines
\put(1931.500,1847.929){\arc{394.717}{2.4948}{6.9299}}
\blacken\thicklines
\path(2091.033,1869.097)(2089.000,1729.000)(2162.998,1847.977)(2091.033,1869.097)
\thinlines
\put(7466.500,1885.929){\arc{394.717}{2.4948}{6.9299}}
\blacken\thicklines
\path(7626.033,1907.097)(7624.000,1767.000)(7697.998,1885.977)(7626.033,1907.097)
\thinlines
\put(6401.500,1870.929){\arc{394.717}{2.4948}{6.9299}}
\blacken\thicklines
\path(6561.033,1892.097)(6559.000,1752.000)(6632.998,1870.977)(6561.033,1892.097)
\thinlines
\put(8552,1460){\ellipse{630}{630}}
\put(9673,1467){\ellipse{630}{630}}
\put(6402,1452){\ellipse{630}{630}}
\put(7459,1470){\ellipse{630}{630}}
\put(10787,1459){\ellipse{630}{630}}
\put(10788,1457){\ellipse{540}{540}}
\put(847,1435){\ellipse{630}{630}}
\put(1930,1430){\ellipse{630}{630}}
\put(4207,1449){\ellipse{630}{630}}
\put(3065,1466){\ellipse{630}{630}}
\path(2172,1654)(2802,1654)
\blacken\thicklines
\path(2667.000,1616.500)(2802.000,1654.000)(2667.000,1691.500)(2667.000,1616.500)
\thinlines
\path(1047,1661)(1677,1661)
\blacken\thicklines
\path(1542.000,1623.500)(1677.000,1661.000)(1542.000,1698.500)(1542.000,1623.500)
\thinlines
\path(7721,1685)(8306,1685)
\blacken\thicklines
\path(8171.000,1647.500)(8306.000,1685.000)(8171.000,1722.500)(8171.000,1647.500)
\thinlines
\path(8817,1677)(9402,1677)
\blacken\thicklines
\path(9267.000,1639.500)(9402.000,1677.000)(9267.000,1714.500)(9267.000,1639.500)
\thinlines
\path(6627,1692)(7212,1692)
\blacken\thicklines
\path(7077.000,1654.500)(7212.000,1692.000)(7077.000,1729.500)(7077.000,1654.500)
\thinlines
\path(12,1452)(499,1452)
\blacken\thicklines
\path(379.000,1422.000)(499.000,1452.000)(379.000,1482.000)(379.000,1422.000)
\thinlines
\path(3372,1459)(3867,1459)
\blacken\thicklines
\path(3732.000,1421.500)(3867.000,1459.000)(3732.000,1496.500)(3732.000,1421.500)
\thinlines
\path(2810,1279)(2225,1279)
\blacken\thicklines
\path(2360.000,1316.500)(2225.000,1279.000)(2360.000,1241.500)(2360.000,1316.500)
\thinlines
\path(10002,1467)(10497,1467)
\blacken\thicklines
\path(10362.000,1429.500)(10497.000,1467.000)(10362.000,1504.500)(10362.000,1429.500)
\thinlines
\path(4519,1467)(6042,1467)
\blacken\thicklines
\path(5907.000,1429.500)(6042.000,1467.000)(5907.000,1504.500)(5907.000,1429.500)
\thinlines
\path(2922,1189)(2921,1188)(2919,1186)
	(2915,1183)(2909,1178)(2900,1171)
	(2888,1162)(2874,1150)(2857,1136)
	(2836,1120)(2813,1102)(2788,1083)
	(2760,1063)(2730,1041)(2699,1019)
	(2666,996)(2632,974)(2596,951)
	(2559,929)(2521,907)(2482,886)
	(2442,866)(2400,846)(2356,828)
	(2310,811)(2262,795)(2213,780)
	(2160,767)(2106,757)(2049,748)
	(1991,742)(1932,739)(1869,739)
	(1809,743)(1750,750)(1695,759)
	(1643,771)(1594,785)(1547,800)
	(1503,817)(1462,835)(1422,854)
	(1384,875)(1348,896)(1314,918)
	(1280,940)(1249,963)(1218,985)
	(1190,1007)(1163,1029)(1139,1049)
	(1116,1068)(1097,1085)(1079,1100)
	(1065,1113)(1054,1124)(1045,1132)(1032,1144)
\blacken\thicklines
\path(1156.634,1079.987)(1032.000,1144.000)(1105.763,1024.877)(1156.634,1079.987)
\thinlines
\path(9717,1159)(9716,1158)(9715,1157)
	(9711,1155)(9706,1151)(9699,1145)
	(9689,1137)(9676,1128)(9660,1116)
	(9642,1102)(9620,1086)(9596,1069)
	(9569,1049)(9539,1028)(9507,1005)
	(9472,982)(9436,957)(9397,932)
	(9357,906)(9316,880)(9272,854)
	(9228,828)(9182,802)(9135,777)
	(9087,752)(9038,728)(8987,705)
	(8934,682)(8880,661)(8824,641)
	(8766,622)(8706,604)(8643,587)
	(8578,573)(8511,560)(8441,549)
	(8369,540)(8294,533)(8219,530)
	(8142,529)(8066,532)(7990,537)
	(7917,545)(7846,556)(7777,568)
	(7712,583)(7649,599)(7589,617)
	(7531,637)(7475,657)(7422,679)
	(7370,701)(7321,725)(7273,749)
	(7227,775)(7181,800)(7138,827)
	(7095,854)(7054,881)(7015,908)
	(6976,935)(6939,961)(6904,988)
	(6871,1013)(6839,1038)(6810,1061)
	(6783,1083)(6758,1103)(6736,1121)
	(6717,1137)(6700,1152)(6686,1164)
	(6675,1174)(6666,1181)(6659,1187)(6649,1196)
\blacken\thicklines
\path(6774.431,1133.563)(6649.000,1196.000)(6724.259,1077.816)(6774.431,1133.563)
\thinlines
\path(9492,1196)(9491,1195)(9489,1193)
	(9485,1190)(9478,1184)(9469,1177)
	(9457,1167)(9442,1155)(9425,1141)
	(9404,1125)(9381,1107)(9356,1088)
	(9328,1068)(9299,1047)(9268,1026)
	(9235,1005)(9201,983)(9166,963)
	(9130,943)(9091,923)(9052,905)
	(9010,887)(8967,871)(8921,856)
	(8872,842)(8821,831)(8767,821)
	(8711,813)(8652,809)(8592,807)
	(8532,809)(8473,813)(8417,821)
	(8362,831)(8311,842)(8262,856)
	(8216,871)(8172,887)(8130,905)
	(8090,923)(8052,943)(8015,963)
	(7979,983)(7945,1005)(7912,1026)
	(7880,1047)(7850,1068)(7823,1088)
	(7797,1107)(7773,1125)(7752,1141)
	(7734,1155)(7719,1167)(7707,1177)
	(7698,1184)(7684,1196)
\blacken\thicklines
\path(7810.904,1136.615)(7684.000,1196.000)(7762.095,1079.671)(7810.904,1136.615)
\put(6319,1380){\makebox(0,0)[lb]{\smash{{\SetFigFont{9}{10.8}{\rmdefault}{\mddefault}{\updefault}$1$}}}}
\put(7377,1387){\makebox(0,0)[lb]{\smash{{\SetFigFont{9}{10.8}{\rmdefault}{\mddefault}{\updefault}$2$}}}}
\put(8486,1380){\makebox(0,0)[lb]{\smash{{\SetFigFont{9}{10.8}{\rmdefault}{\mddefault}{\updefault}$3$}}}}
\put(9589,1388){\makebox(0,0)[lb]{\smash{{\SetFigFont{9}{10.8}{\rmdefault}{\mddefault}{\updefault}$4$}}}}
\put(10706,1388){\makebox(0,0)[lb]{\smash{{\SetFigFont{9}{10.8}{\rmdefault}{\mddefault}{\updefault}$5$}}}}
\put(1835,1361){\makebox(0,0)[lb]{\smash{{\SetFigFont{9}{10.8}{\rmdefault}{\mddefault}{\updefault}$q_2$}}}}
\put(2300,1789){\makebox(0,0)[lb]{\smash{{\SetFigFont{9}{10.8}{\familydefault}{\mddefault}{\updefault}$a,b$}}}}
\put(1790,829){\makebox(0,0)[lb]{\smash{{\SetFigFont{9}{10.8}{\familydefault}{\mddefault}{\updefault}$a$}}}}
\put(596,2179){\makebox(0,0)[lb]{\smash{{\SetFigFont{9}{10.8}{\familydefault}{\mddefault}{\updefault}$b,d$}}}}
\put(1865,2156){\makebox(0,0)[lb]{\smash{{\SetFigFont{9}{10.8}{\familydefault}{\mddefault}{\updefault}$d$}}}}
\put(1295,1774){\makebox(0,0)[lb]{\smash{{\SetFigFont{9}{10.8}{\familydefault}{\mddefault}{\updefault}$a$}}}}
\put(3485,1602){\makebox(0,0)[lb]{\smash{{\SetFigFont{9}{10.8}{\familydefault}{\mddefault}{\updefault}$d$}}}}
\put(2420,1369){\makebox(0,0)[lb]{\smash{{\SetFigFont{9}{10.8}{\familydefault}{\mddefault}{\updefault}$b$}}}}
\put(6754,1827){\makebox(0,0)[lb]{\smash{{\SetFigFont{9}{10.8}{\familydefault}{\mddefault}{\updefault}$a$}}}}
\put(10077,1587){\makebox(0,0)[lb]{\smash{{\SetFigFont{9}{10.8}{\familydefault}{\mddefault}{\updefault}$d$}}}}
\put(7789,1835){\makebox(0,0)[lb]{\smash{{\SetFigFont{9}{10.8}{\familydefault}{\mddefault}{\updefault}$a,b$}}}}
\put(8878,1835){\makebox(0,0)[lb]{\smash{{\SetFigFont{9}{10.8}{\familydefault}{\mddefault}{\updefault}$a,b$}}}}
\put(7819,109){\makebox(0,0)[lb]{\smash{{\SetFigFont{9}{10.8}{\rmdefault}{\mddefault}{\updefault}$\cR_5(a,b,d)$}}}}
\put(1752,117){\makebox(0,0)[lb]{\smash{{\SetFigFont{9}{10.8}{\rmdefault}{\mddefault}{\updefault}$\cR_4(a,b,d)$}}}}
\put(7355,2208){\makebox(0,0)[lb]{\smash{{\SetFigFont{9}{10.8}{\familydefault}{\mddefault}{\updefault}$d$}}}}
\put(8436,2224){\makebox(0,0)[lb]{\smash{{\SetFigFont{9}{10.8}{\familydefault}{\mddefault}{\updefault}$d$}}}}
\put(6108,2216){\makebox(0,0)[lb]{\smash{{\SetFigFont{9}{10.8}{\familydefault}{\mddefault}{\updefault}$b,d$}}}}
\put(10399,2216){\makebox(0,0)[lb]{\smash{{\SetFigFont{9}{10.8}{\familydefault}{\mddefault}{\updefault}$a,b,d$}}}}
\put(3822,2201){\makebox(0,0)[lb]{\smash{{\SetFigFont{9}{10.8}{\familydefault}{\mddefault}{\updefault}$a,b,d$}}}}
\put(5103,1609){\makebox(0,0)[lb]{\smash{{\SetFigFont{9}{10.8}{\familydefault}{\mddefault}{\updefault}$\eps$}}}}
\put(7587,702){\makebox(0,0)[lb]{\smash{{\SetFigFont{9}{10.8}{\familydefault}{\mddefault}{\updefault}$a$}}}}
\put(8886,1001){\makebox(0,0)[lb]{\smash{{\SetFigFont{9}{10.8}{\familydefault}{\mddefault}{\updefault}$b$}}}}
\put(694,1369){\makebox(0,0)[lb]{\smash{{\SetFigFont{9}{10.8}{\rmdefault}{\mddefault}{\updefault}$q_1$}}}}
\put(4084,1382){\makebox(0,0)[lb]{\smash{{\SetFigFont{9}{10.8}{\rmdefault}{\mddefault}{\updefault}$q_4$}}}}
\thinlines
\put(4211.500,1870.929){\arc{394.717}{2.4948}{6.9299}}
\blacken\thicklines
\path(4371.033,1892.097)(4369.000,1752.000)(4442.998,1870.977)(4371.033,1892.097)
\end{picture}
}
\end{center}
\caption{Right-ideal witnesses for product.} 
\label{fig:product}
\end{figure}
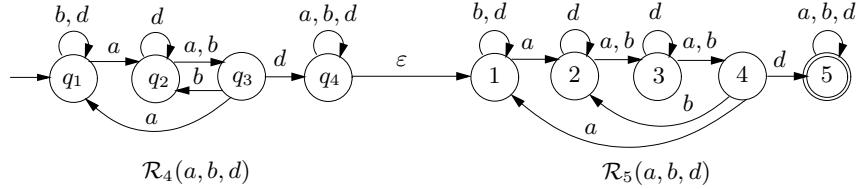

\begin{theorem} [Product]
\label{thm:product}
For $m\ge 1$, $n\ge 2$, the complexity of the product $R_m(a,b,d)\cdot R_n(a,b,d)$ is $m+2^{n-2}$.
\end{theorem}
\begin{proof}
It was shown in~\cite{BJL13} that $m+2^{n-2}$ is an upper bound on the complexity of the product of two right ideals. To prove this bound is met, 
we apply the subset construction to $\cP$ to obtain a DFA $\cD$ for $R_mR_n$.
The states of $\cD$ are subsets of $Q'_m\cup Q_n$.
We prove that all states of the form $\{q_i\}$, $i=1,\ldots,m-1$ and all states of the form
$\{q_{m},1\}\cup S$, where $S\subseteq Q_n\setminus\{1,n-1\}$, and state 
$\{q_{m},1,n\}$ are reachable, for a total of $m+2^{n-2}$ states.

State $\{q_1\}$ is the initial state, and  $\{q_i\}$ is reached by $a^{i-1}$ for $i=2,\ldots,m-1$.
Also, $\{q_{m},1\}$ is reached by $a^{m-2}d$. 
States $q_{m}$ and 1 are present in every subset reachable from now on. 
 By applying $ab^{j-1}$ to $\{q_{m},1\}$ we reach $\{q_{m},1,j\}$; hence all subsets 
 $\{q_{m},1\}\cup S$ with $|S|=1$ are reachable.
 Assume now that we can reach all sets $\{q_{m},1\}\cup S$ with $|S|=k$, and 
 suppose that we want to reach $\{q_{m},1\}\cup T$ with $T=\{i_0,i_1,\ldots,i_k\}$
 with $2\le i_0<i_1<\cdots <i_k\le n-1$.
 Start with $S=\{i_1-i_0, \ldots, i_k-i_0\}$ and apply $ab^{i_0-1}$.
 Finally, to reach $\{q_{m},1,n\}$, start with $\{q_{m},1,n-1\}$ and apply $d$.
 
If $1\le  i<j\le m-1$, then state $\{q_i\}$ is distinguishable from $\{q_j\}$ by 
$a^{m-1-j}da^{n-1}d$.
 Also, state $i\in Q_n$ with $2\le j\le n-1$ accepts $a^{n-1-j}d$ and no other state $j\in Q_n$ with $2\le j\le n-1$ accepts this word. 
 Hence, if $S,T \subseteq Q_n\setminus\{1,n-1\}$ and $S\ne T$, then 
 $\{q_{m},1\}\cup S$ and $\{q_{m},1\}\cup T$ are distinguishable.
 State $\{q_k\}$ with $2\le k\le m-1$ is distinguishable from state $\{q_{m},1\}\cup S$
 because there is a word with a single $d$ that is accepted from $\{q_{m},1\}\cup S$
 but no such word is accepted by $\{q_k\}$. Hence all the non-final states are distinguishable, and $\{q_{m},1,n\}$ is the only final state.
\qed
\end{proof}
\section{Conclusions}
We have shown that there is a stream of regular right ideal that acts as universal witness for all common operations.

\end{document}